\documentclass[12pt, draftclsnofoot, onecolumn]{IEEEtran}
\usepackage[utf8]{inputenc}

\usepackage{amssymb, amsmath, amsfonts, amsthm,color}
\usepackage{caption, subcaption, cite}
\usepackage{graphicx, cite, mathtools}
\usepackage{dsfont}
\newtheorem{theorem}{Theorem}

\newtheorem{lemma}{Lemma}

\newtheorem{proposition}{Proposition}

\DeclareMathOperator*{\argmax}{arg\,max}

\normalsize

\usepackage{mathrsfs}

\begin{document}

\title{Fluid Antenna with Linear MMSE Channel Estimation for Large-Scale Cellular Networks}
\author{Christodoulos Skouroumounis, \IEEEmembership{Member},
	and~Ioannis Krikidis, \IEEEmembership{Fellow}% <-this % stops a space
	\thanks{Christodoulos Skouroumounis and Ioannis Krikidis are with the IRIDA Research Centre for Communication Technologies, Department of Electrical and Computer Engineering, University of Cyprus, Cyprus, e-mail:\{cskour03, krikidis\}@ucy.ac.cy.}}

\maketitle

\begin{abstract}
The concept of reconfigurable fluid antennas (FA) is a potential and promising solution to enhance the spectral efficiency of wireless communication networks. Despite their many advantages, FA-enabled communications have limitations as they require an enormous amount of spectral resources in order to select the most desirable position of the radiating element from a large number of prescribed locations. In this paper, we present an analytical framework for the outage performance of large-scale FA-enabled communications, where all user equipments (UEs) employ circular multi-FA array. In contrast to existing studies, which assume perfect channel state information, the developed framework accurately captures the channel estimation errors on the performance of the considered network deployments. In particular, we focus on the limited coherence interval scenario, where a novel sequential linear minimum mean-squared error (LMMSE)-based channel estimation method is performed for only a very small number of FA ports. Next, for the communication of each BS with its associated UE, a low-complexity port-selection technique is employed, where the port that provides the highest signal-to-interference-plus-noise-ratio is selected among the ports that are estimated to provide the strongest channel from each FA. By using stochastic geometry tools, we derive both analytical and closed-form expressions for the outage probability, highlighting the impact of channel estimation on the performance of FA-based UEs. Our results reveal the trade-off imposed between improving the network's performance and reducing the channel estimation quality, indicating new insights for the design of FA-enabled communications.
\end{abstract}


\IEEEpeerreviewmaketitle
\newpage

\section{Introduction}
The increasingly demanding objectives for sixth generation (6G) wireless communication systems have spurred recent research activities on novel transceiver hardware architectures and relevant communication algorithms. Towards this direction, the concept of high-performance fluid antennas (FAs) has envisioned as a promising technology in future communication devices due to a range of attractive features such as conformability, flexibility, and reconfigurability \cite{AKY}. More specifically, FAs consist of liquid radiating elements (e.g., mercury, eutectic gallium indium, galinstan, etc.), which are contained in a dielectric holder and can thus flow in different locations (i.e., a set of predefined ports) within its topological boundaries with the assistance of a dedicated microelectromechanical system (MEMS)\cite{HUA}. Therefore, FAs are capable to reversibly re-configure their physical configuration (i.e., size, shape and feeding) as well as their electrical properties (e.g., resonant frequency, bandwidth), providing a new degree of freedom in the design of wireless communication systems \cite{TAR}. 

As a result, FA-enabled communications has recently attracted extensive attention by the research community and the industry and has already been investigated under various different communication scenarios. In particular, the wide range of functionalities attain by the employment of reconfigurable FAs is investigated, e.g. \cite{MOR,DEY,BOR,SIN,SHE}, while a contemporary survey on this topic can be found in \cite{HUA}. From an information/communication theoretical standpoint, the concept of FA is examined for point-to-point communication systems, where the achieved performance in terms of outage probability \cite{WON1} and ergodic capacity \cite{WON2} is assessed under spatially correlated channels. It is revealed that a single FA with half-wavelength or less separation between ports can attain capacity and outage performance similar to the conventional multi-antenna maximum ratio combining (MRC) system, if the number of ports is sufficiently large. The aforementioned works are further extended into the context of multi-user communications in \cite{WON3}, where the authors propose a mathematical framework that takes into consideration multiple pairs of transmitters and FA-based receivers. More specifically, the outage performance, the achievable rate, and the multiplexing gain of the considered topology are evaluated, demonstrating that the overall network performance improves with the number of ports at each receiver. Moreover, the problem of port selection in the context of FA-enabled communications is investigated in \cite{CHA}. Specifically, the authors investigate several port selection algorithms by exploiting both machine learning and analytical approximation tools for a communication scenario where a transmitter observes only a limited number of ports to reduce complexity. Although the above works shed light on the performance experienced by user equipments (UEs) with a single FA, the concept of multi-FA UEs, aiming to boost the achieved diversity gain and thereby enhance the network capacity, is missing from the literature.

To unleash the full potential of FA-enabled communications, accurate channel estimation is an essential prerequisite. Nevertheless, the aforementioned studies assume that the FA-enabled communications have perfect knowledge of the channel state information (CSI) for all FA-associated links. In practice, such CSI needs to be acquired in each channel coherence interval at the cost of a channel training overhead that escalates with the number of FA ports \cite{YAZ}. In the context of a limited coherence interval scenario, this training period inevitably leads in a reduced data transmission duration and, consequently, in a decreased overall network performance. Therefore, the concept of channel estimation in large-scale FA-enabled communications, triggers a non-trivial trade-off between the channel training duration and the overall network performance. The impact of channel estimation on the network performance has attracted substantial attention by the research community and has already been investigated under different communication scenarios, including simple network settings \cite{LOU,KAY} and cellular networks \cite{JOS,YIN}. Nevertheless, a major limitation of the above-mentioned works is the adopted deterministic network configuration, which does not captures the irregularity associated with the actual deployments of cellular networks. In such a context, system-level performance evaluation will be very important to formulate relevant insights into trade-offs that govern such a complex system. Over the past decade, stochastic geometry (SG) has emerged as a powerful mathematical, which captures the random nature of large-scale networks and permits the analytical characterization of numerous performance metrics \cite{SGbook}. Even though the impact of channel estimation on the network performance has been widely studied in small-scale networks, such as single-cell scenarios, there is minimal research on the effects on large-scale networks. The authors in \cite{WU1}, evaluated the impact of channel estimation in the context of point-to-point single-input single-output ad-hoc systems, by using linear minimum mean square error (LMMSE) channel estimation. The coverage probability and the impact of channel estimation on the performance of random networks have been studied in \cite{WU2}, capturing the dependence of the optimal training-pilot length on the ratio between the receiver and transmitter densities. In the context of FA systems, in \cite{SKO}, the authors assess the effect of FA technology on large-scale cellular networks, and study the trade-off imposed by the channel estimation on the outage performance; it is unveiled that an optimal number of FAs' ports maximizes the network performance with respect to the LMMSE-based channel estimation process.

Motivated by the above, in this work, we propose a novel LMMSE-based port selection (PS) technique and investigate the achieved performance of large-scale FA-enabled cellular networks under a limited coherence interval scenario, where all UEs employ a circular multi-FA array. By taking into account spatial randomness, we provide a rigorous mathematical framework to analyze the performance of the FA-based UEs employing the PS technique, in terms of outage probability. More specifically, the main contributions of this paper are summarized as follows:
\begin{itemize}
	\item An analytical framework is proposed based on SG, which comprises the co-design of FA technology and homogeneous cellular networks, shedding light on the modeling and analysis of large-scale FA-enabled communications. Through this paper, we extend the work presented in \cite{SKO} by considering multiple FAs at the UEs, aiming at elevating the spatial diversity gain and thus improving the achieved spectral efficiency. Based on the developed mathematical framework, the performance of the considered network deployments is assessed in terms of outage probability under a limited coherence interval scenario.
	\item Building on the developed mathematical framework, we propose two novel and low-complexity schemes, namely skipped-enabled LMMSE-based channel estimation (SeCE) technique and port-selection (PS) scheme. Initially, by leveraging the existence of strong spatial correlation between adjacent FA's ports, the proposed SeCE scheme leads to a LMMSE-based channel estimation process of only a subset of ``selected'' ports from each FA, thereby reducing the signalling overhead at the cost of reduced signal-to-interference-plus-noise-ratio (SINR). Consequently, the employment of the proposed SeCE technique triggers a non-trivial trade-off between mitigating the signaling overhead and reducing the observed SINR. Furthermore, the proposed PS scheme is based on a two-stage procedure. In the first stage, a set of ``candidate'' ports is defined consisting of the ``selected'' port that is estimated to provide the strongest channel from each FA while, at the second stage, the port that provides the highest SINR is selected among the set of ``candidate'' ports. The proposed two-stage PS scheme reduces the signaling overhead and thus is promising for practical and low-complexity implementations. 
	\item By using SG tools, an analytical expression for the outage performance achieved by the considered network deployments is derived. Moreover, under specific practical assumptions, closed-form expressions for the upper and lower bound of the conditional outage performance are derived. These closed-form expressions provide a quick and convenient methodology to evaluate the system's performance and obtain insights into key design parameters. Our results unveil that an optimal number of FAs and FAs' ports maximize the network performance with respect to the LMMSE-based channel estimation technique. Finally, compared to the conventional schemes, where channel estimation is performed for all FAs' ports, the proposed scheme significantly enhances the achieved network performance, especially for FAs with a large number of ports.
\end{itemize}

\textit{Organization:} Section II introduces the system model together with the network topology, FA, and channel models. The low-complexity SeCE and PS schemes are proposed in Section III, and the analytical along with the asymptotic expressions for the performance achieved by the large-scale FA-based cellular networks are evaluated in Section IV. Simulation results are presented in Section V, followed by our conclusions in Section VII. 
\begin{table*}[t]\centering
	\caption{Summary of Notations}\label{Table1}
	\scalebox{0.85}{
		\begin{tabular}{| l | l || l | l |}\hline
			\textbf{Notation} 		& \textbf{Description} 									& \textbf{Notation} 						& \textbf{Description}\\\hline
			$\Phi,\lambda_b$ 		& PPP of BSs of density $\lambda_b$ 					&  $W_c, T_c$								& Coherence bandwidth and time  \\\hline
			$\Psi$, $\lambda_u$ 	& Point process of UEs of density $\lambda_u$  			&  $L_e, L_t$								& Channel estimation and data transmission period  \\\hline
			$M$,$\mathcal{M}$ 		& Number and set of FAs equipped by each UE 			&  $\nu$									& Number of skipped FA's ports \\\hline
			$N,\mathcal{N}$ 		& Number and set of FA's ports 							&  $N',\mathfrak{N}$						& Number and set of ``selected'' ports  \\\hline
			$\kappa$, $\lambda$ 	& Scaling constant and communication wavelength 		&  $l_s$									& Switching channel uses  \\\hline
			$d_i$ 					& Displacement of the $i$-th port  						&  $\Delta$									& Number of pilot-training symbols for each port\\\hline
			$u,\delta(x)$ 			& Average velocity and delay of fluid metal  	&  $\mathcal{I}$							& Instantaneous power of multi-user interference \\\hline
			$r_i(\rho)$ 			& Distance between the $i$-th port and the serving BS 	&  $f_\mathcal{I}(\cdot),\varpi,\varrho$	& Gamma distribution with parameters $\varpi$ and $\varrho$ \\\hline
			$P$ 					& Transmission power  									&  $\mathcal{C}$							& Set of ``candidate'' ports \\\hline
			$\ell(r)$ 				& Large scale path loss 								&  SINR$_i$									& SINR observed at the $i$-th port \\\hline
			$\mu_i$ 				& Autocorrelation parameter 							&  $\varepsilon$							& Transmit signal to noise ratio \\\hline
			$L_c$					& Channel uses in each coherence interval 				&  $\sigma_e^2|\rho$						& Variance of the channel estimation error \\\hline
	\end{tabular}}
\end{table*}
\section{System Model}\label{SystemModel}
In this section, we provide details of the considered system model. A list of the main mathematical notations is presented in Table \ref{Table1}.
\subsection{Network topology}
\begin{figure}
	\centering\includegraphics[width=0.55\linewidth]{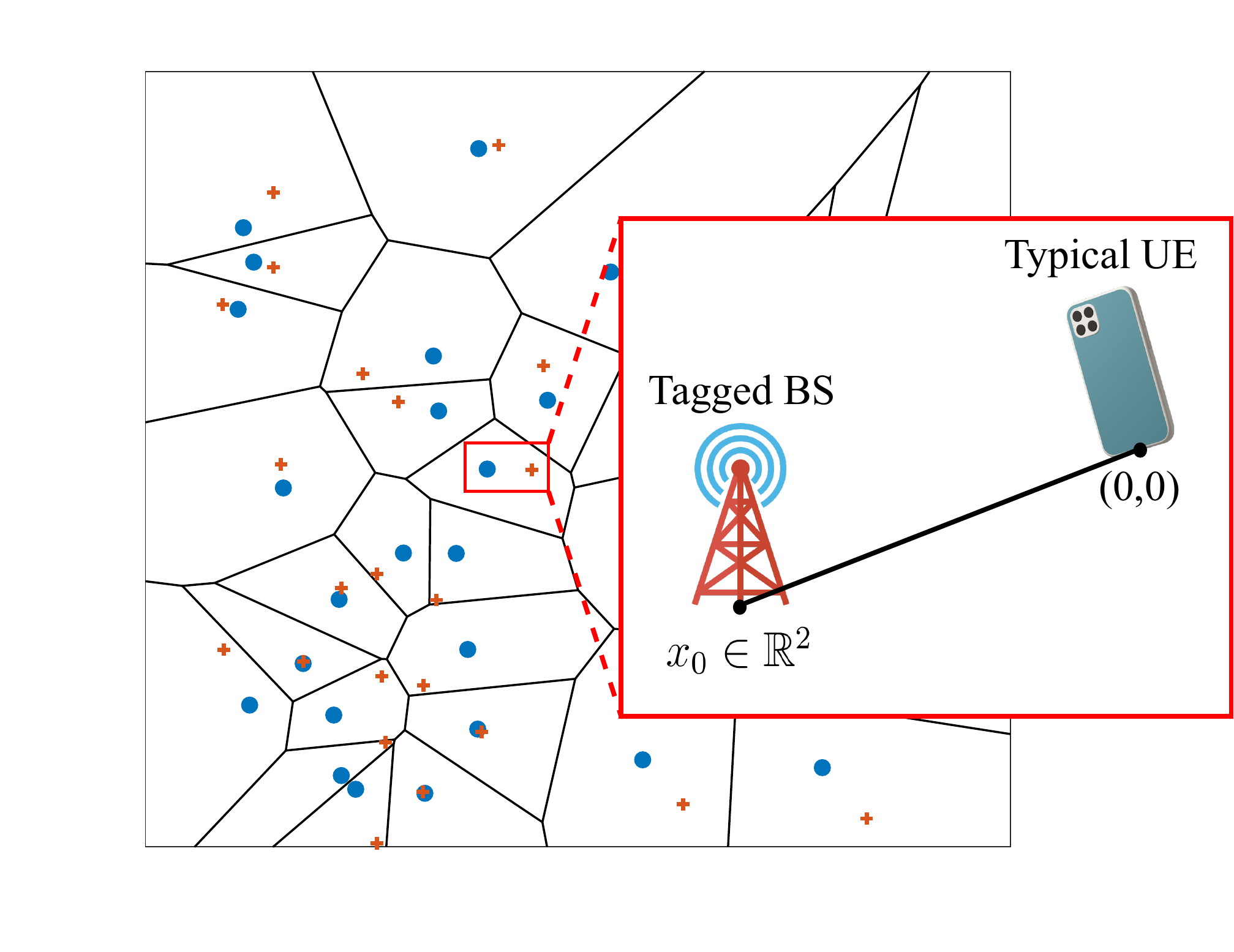}\vspace{-20pt}
	\caption{The Voronoi tessellation of a large-scale FA-enabled cellular network, where BSs and UEs are represented by circles and crosses, respectively.}\label{Voronoi}\vspace{-20pt}
\end{figure}
We consider a single-tier downlink cellular network. The locations of the BSs are modeled as points of a homogeneous Poisson point process (PPP), denoted as $\Phi=\left\lbrace x_i\in\mathbb{R}^2,i\in\mathbb{N}^+\right\rbrace $, of spatial density $\lambda_b$ BS/$\rm{m^2}$. Furthermore, the locations of the UEs follow an arbitrary independent point process $\Psi$ with spatial density $\lambda_u\gg  \lambda_b$. With the aim of facilitating the analysis and reducing the required computational complexity of such network deployments, we consider different practical assumptions in order to derive ``close'' to closed-form expressions for the achieved network performance and thus obtain insights into how key system parameters affect the performance. More specifically, we assume that all BSs are equipped with a single omnidirectional antenna, while all UEs are equipped with $M$ FAs \footnote{Although the employment of a large number of FAs could be implemented in some end-user devices e.g., laptops, tablets, and handhelds, the limited size of the majority of UEs raises some practicality and feasibility issues in the context of conventional sub-$6$ GHz communication networks. Fortunately, when moving to higher frequencies e.g., millimeter-wave or terahertz, the small wavelength associated with the higher frequency signals can be exploited, enabling a large FA array to be packed in a small physical size \cite{BOG}.}, where $M\in\mathbb{N}^+$ (detailed description in Section \ref{FAmodel}). Regarding the adopted multiple access scheme, we assume the employment of an orthogonal multiple access technique, e.g. time division multiple access, such that each BS serves a single UE at the time without the existence of intra-cell interference. Without loss of generality and by following Slivnyak's theorem \cite{SGbook}, the analysis concerns the typical UE located at the origin but the results hold for all UEs of the network. We consider the nearest-BS association rule i.e., the typical UE at the origin communicates with its closest BS located at $x_{\rm 0}\in\mathbb{R}^2$, referred as \textit{tagged BS}, and its link with the typical UE is denoted as \textit{typical link} (see Fig. \ref{Voronoi}).

\subsection{Fluid antenna model}\label{FAmodel}
\begin{figure}
	\centering
	\begin{subfigure}{.4\textwidth}
		\centering
		\includegraphics[width=\linewidth]{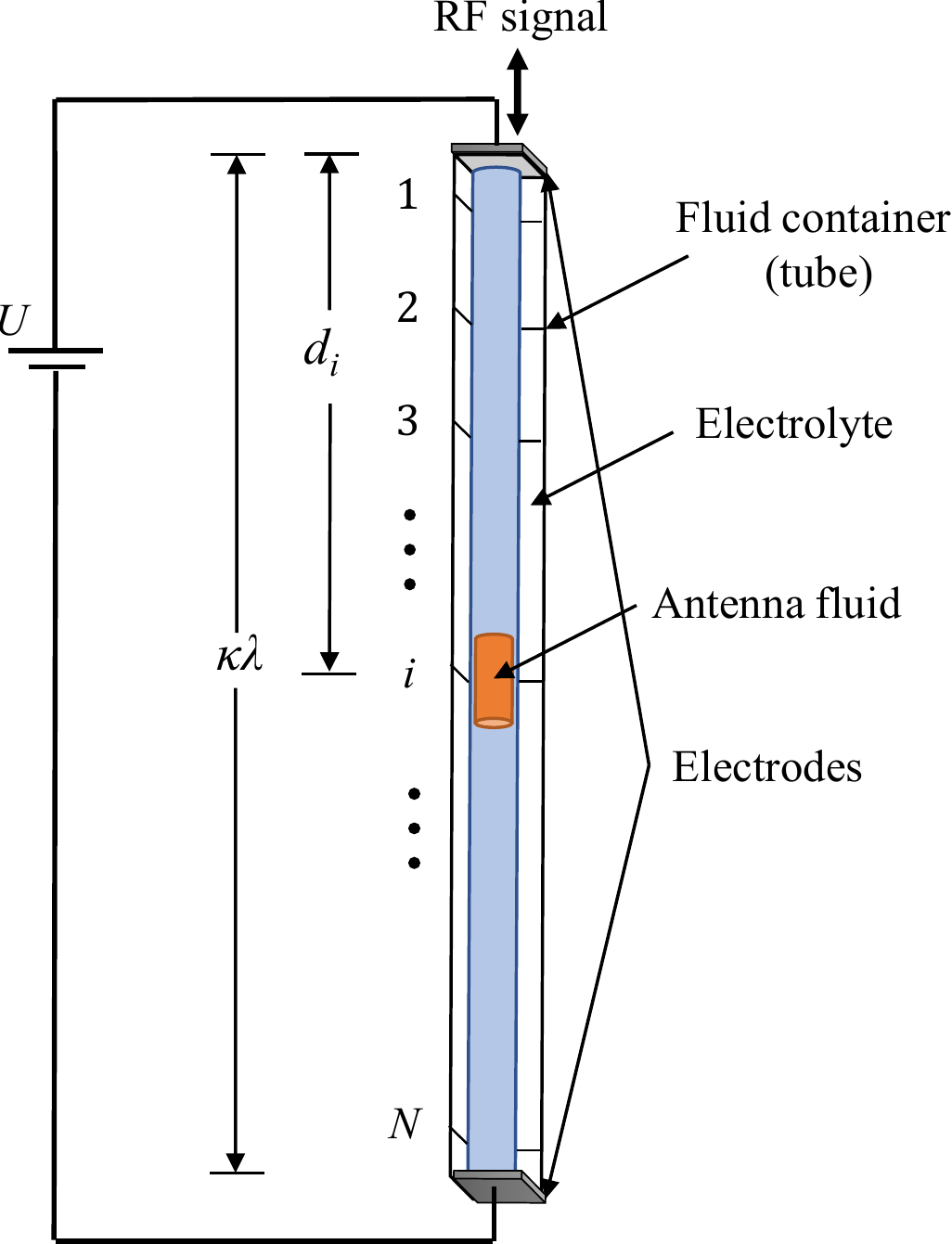}
		\caption{A potential FA architecture.}\label{FA}
	\end{subfigure}%
	\begin{subfigure}{.6\textwidth}
		\centering
		\includegraphics[width=.87\linewidth]{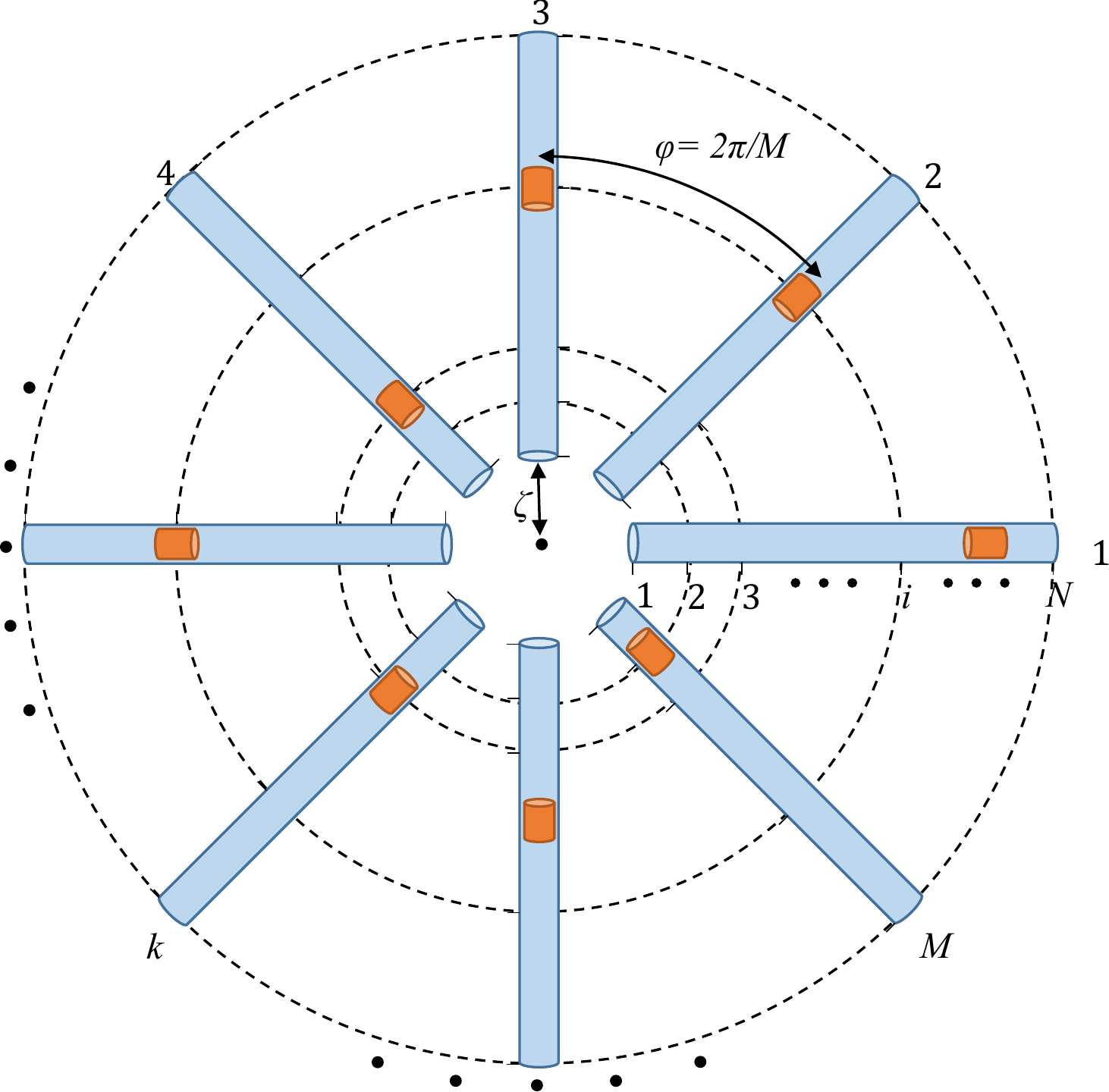}
		\caption{The investigated circular $M$-FA array equipped by each UE.}\label{MultiFA}
	\end{subfigure}
	\caption{The considered FA concept.}
	\label{fig:test}
\end{figure}
Fig. \ref{FA} illustrates the considered architecture of a fluid-based reconfigurable antenna. In the figure, a drop of liquid metal (e.g., EGaIn) is placed in a tube-like linear microchannel or capillary filled with an electrolyte, within which the fluid is free to move. In particular, the location of the antenna (i.e., fluid metal) can be promptly switched to one of the $N$ preset locations (also known as ``ports''), that are evenly distributed along the linear dimension of a FA, $\kappa\lambda$, where $\lambda$ is the wavelength of communication and $\kappa$ is a scaling constant. An abstraction of the FA concept is considered, where an antenna at a given location is treated as an ideal point antenna \cite{WON1}. By aiming to ease the mathematical analysis, the first port of a FA is treated as an auxiliary reference port \cite{WON1}. Thus, as demonstrated in Fig. \ref{FA}, the distance between the reference port and the $i$-th port can be measured as follows
\begin{equation}
d_i = \left(\frac{i-1}{N-1}\right)\kappa\lambda,\ \forall\ i\in\mathcal{N},
\end{equation}
where $\mathcal{N}=\{1,2,\dots,N\}$.
%\begin{figure}
%	\centering\includegraphics[width=.5\linewidth]{FAsingleNew.pdf}
%	\caption{The considered FA concept.}\label{FA}
%\end{figure}
%\begin{figure}
%	\centering\includegraphics[width=.5\linewidth]{MultiFA.pdf}
%	\caption{The considered FA concept.}\label{MultiFA}
%\end{figure}

We assume the existence of a MEMS at each UE, which is responsible for the flow motion of a fluid metal within an electrolyte-filled capillary. More specifically, the flow motion of the fluid metal is induced by the application of a voltage gradient along the FA, as shown in Fig. \ref{FA}, as a result of the electrocapillary effect \cite{LEE}. This particular phenomenon, called continuous electrowetting \cite{GOU}, is an electrical analogy of the well-known Marangoni effect \cite{GOU}. Therefore, according to Hagen–Poiseuille equation \cite{LEE}, the average velocity achieved by a fluid metal is given by
\begin{equation}\label{velocity}
u=\frac{q}{6\mu}\frac{D}{L}\Delta\phi,
\end{equation}
where $q$ denotes the initial charge in the electrical double layer for EGaIn in most of the aqueous electrolytes, $\mu$ is the viscosity of EGaIn at $20^oC$, $D$ and $L$ represent the thickness and the length of the fluid metal, respectively. $\Delta\phi$ represents the voltage difference between the two ends of the fluid metal which is much smaller than the externally applied voltage $U$ i.e., $\Delta\phi\ll U$. Therefore, in contrast to existing works that consider high-velocity fluid metals with instantaneous port switching, the assumption of fluid metals with finite velocity requires a non-zero time period for the movement of the fluid metal from a port to another, referred as ``delay''. In particular, the time required (delay) by the fluid metal to move from the $i$-th port to the $j$-th port, is given by
\begin{equation}\label{delay}
\delta(x) =\frac{\kappa\lambda}{u}\left(\frac{x}{N-1}\right),
\end{equation}
where $\{i,j\}\in\mathcal{N}$ and $x=|i-j|$. For example, a fluid metal velocity of $116$ mm/s is obtained by assuming $\Delta\phi=0.1$ V and $L/D=5$ \cite{GOU}, and hence, a delay of $54\ \mu$s is experienced between neighbouring ports (i.e., $x=1$) of a FA architecture with $N=20$, $\kappa = 0.2$, and $\lambda=6$ cm \cite{WON1}.

By aiming at achieving circular symmetry and thereby analytical tractability of our proposed mathematical framework, the geometry of the considered $M$-FA array that is adopted by all UEs is depicted in Fig. \ref{MultiFA}\footnote{The proposed mathematical framework can be applied to more sophisticated FA-array topologies but this topology is sufficient for the purpose of this work.}. More specifically, the $M$ FAs are placed in a circular configuration with $2\pi/M$ angle between the adjacent ones, and we assume that all ports of the $M$ FAs share a common RF chain. Furthermore, we assume that the reference port (i.e., first port) of each FA is located towards the center of the adopted FA array. For simplicity purposes, we assume that the distance between the $M$ reference ports and the center of the considered circular $M$-FA array is zero i.e., $\omega=0$, without loss of generality and accuracy of the proposed mathematical framework. Therefore, the distance of the typical link between the $i$-th port of the $k$-th FA and the tagged BS can be calculated as
\begin{equation}\label{Distance}
	r_i(\rho) = \sqrt{\rho^2+\kappa^2\lambda^2\left(\frac{i-1}{N-1}\right)^2},\ \forall\ i\in\mathcal{N},\ k\in\mathcal{M},
\end{equation}
where  $\mathcal{M}=\{1,2,\dots,M\}$ and $\rho$ represents the distance of the typical link, i.e. $\rho=\|v_0\|$ with probability density function (pdf) $f_R(\rho)$. The above pdf can be derived by differentiating the complementary cumulative distribution function of $R$ i.e., $\mathbb{P}[R>\rho]=\exp\left(-\pi\lambda_b\rho^2\right)$, with respect to $\rho$, that can be calculated as \cite{SGbook}
\begin{equation}\label{ContactPDF}
	f_R(\rho) = \frac{{\rm d}\mathbb{P}[R>\rho]}{{\rm d}\rho}=2\pi\lambda_b\rho\exp\left(-\pi\lambda_b\rho^2\right).
\end{equation}
Note that, owing to the circular symmetry achieved with the proposed $M$-FA array deployment, the distances between the tagged BS and the $i$-th ports of all FAs are equal i.e., $r_i(\rho)$ is independent of $k\in\mathcal{M}$.

\subsection{Channel model}
All wireless signals are assumed to experience both large-scale path-loss effects and small-scale fading. More specifically, all BSs are assumed to use the same transmission power $P$ (dBm). The large-scale attenuation of the transmitted signals follows an unbounded singular path-loss model based on the distance $r$ between a transmitter and a receiver, i.e. $\ell(r)=r^{a}$, where $a>2$ denotes the path-loss exponent. Regarding small-scale fading, we consider a block fading channel model. In other words, the channel remains constant during a coherence time $T_c$, also known as \textit{channel coherence interval}, and evolves independently from block to block. In particular, we assume that the small-scale fading between a transmitter and a receiver follows a circularly symmetric complex Gaussian distribution with zero mean and variance of $\sigma^2$. Hence, the channel's amplitude between the $i$-th port of the typical UE's $k$-th FA and its tagged BS, $\left| g_{{\rm 0}i}^{(k)}\right| $, it is Rayleigh distributed. While such an assumption clearly does not reflect a real network environment, it still enables us to obtain some closed-form expressions which may be used as estimates for more realistic situations. Since the ports of each FA can be arbitrarily close to each other, the channels are considered to be correlated. In particular, the channels observed by the $N$ ports at the $k$-th FA of the typical UE can be evaluated as
\begin{equation}\label{ChannelCorrelation}
	g_{{\rm 0}i}^{(k)}=\begin{cases}
		\sigma \alpha_1+j\sigma \beta_1 &\text{if $i=1$},\\
		\sigma\left(\sqrt{1-\mu_i^2}\alpha_i+\mu_i \alpha_0\right)+j\sigma\left(\sqrt{1-\mu_i^2}\beta_i+\mu_i \beta_0\right) &\text{otherwise},\\
	\end{cases}
\end{equation}
where $i\in\mathcal{N}$, $k\in\mathcal{M}$, $\alpha_1,\dots,\alpha_N,\beta_1,\dots,\beta_N$ are all independent Gaussian random variables with zero mean and variance of $\frac{1}{2}$; $\mu_i$ is the autocorrelation parameter that can be chosen appropriately to determine the channel correlation between the $i$-th and the reference (i.e., first) port of each FA. In particular, we assume that the autocorrelation parameter can be evaluated as \cite{WON1}
\begin{equation}
	\mu_i = \begin{cases}
		0&\text{if $i=1$},\\
		J_0\left(\frac{2\pi(i-1)}{N-1}\kappa\right) &\text{otherwise},\\
\end{cases}
\end{equation}
where $J_0(\cdot)$ is the zero-th Bessel function of the first kind. In addition, we assume that there is no channel correlation between ports that belong into different FAs, as we assume that a $\lambda/2$-separation is ensured between adjacent FAs.
\section{Port Selection with LMMSE Channel Estimation}
\begin{figure}
	\centering\includegraphics[width=0.8\linewidth]{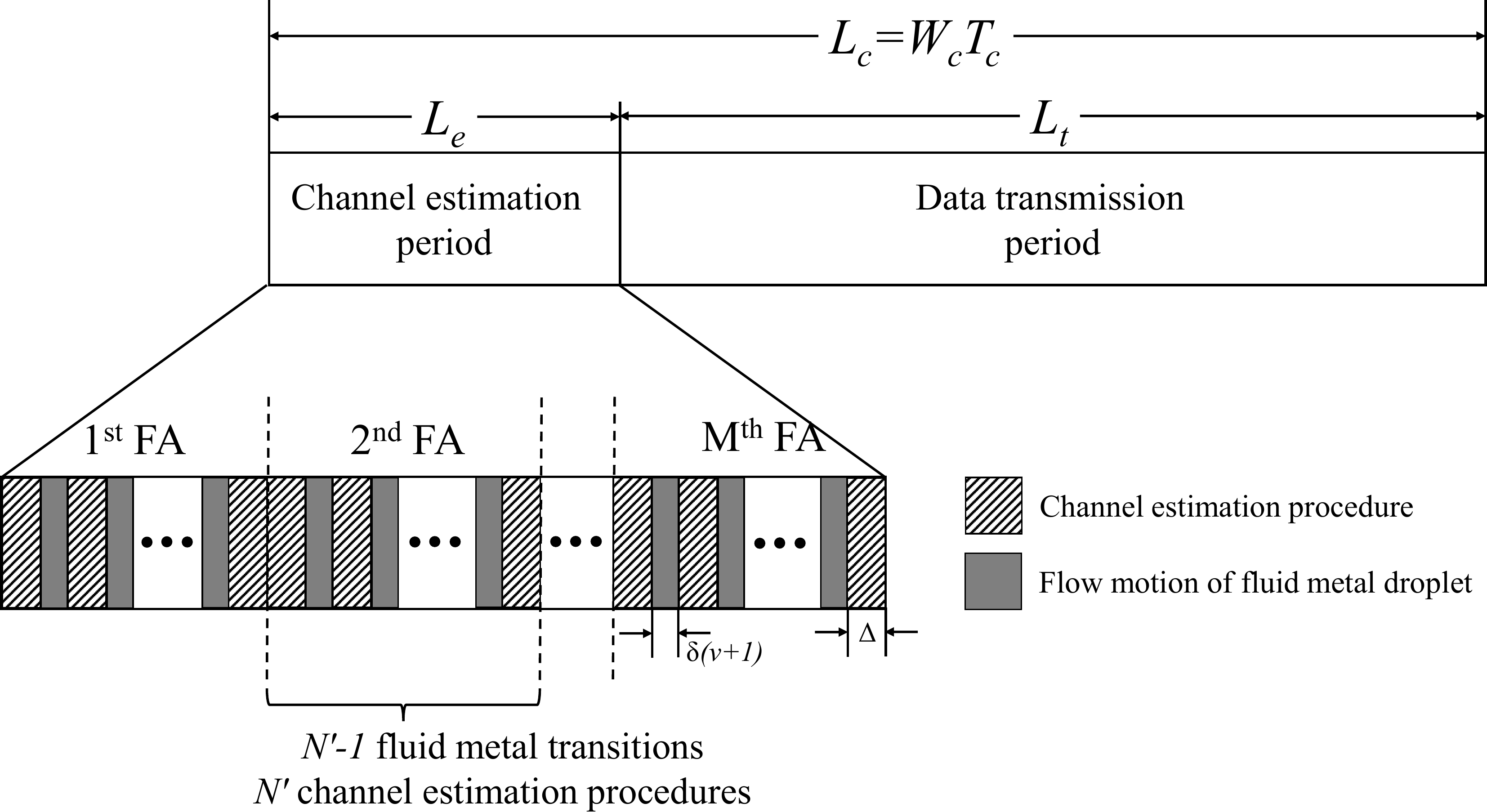}
	\caption{Representation of a block fading channel consisting of channel estimation and data transmission periods.}\label{Time}\vspace{-20pt}
\end{figure}
In this section, we introduce the proposed low-complexity port selection (PS) scheme and the associated skip-enabled LMMSE channel estimation (SeCE) for a scenario with a limited coherence interval. Based on the adopted block fading model, the length of each coherence interval/block ($L_c$ channel uses) is equal to the product of the coherence bandwidth $W_c$ in Hz and the coherence time $T_c$ in s \cite{Tsebook}. Each coherence block is divided into two sub-blocks for channel estimation and data communication, as illustrated in Fig. \ref{Time}. In particular, we consider $L_e$ channel uses per block for pilot-training symbols to enable channel estimation, and $L_t=L_c-L_e$ channel uses are used for downlink (DL) data transmission. Throughout this paper, we will denote by $n$ the channel use during the channel coherence block, i.e. $n\in\{1,\dots,L_c\}$. In the following sections, we elaborate in detail the channel estimation and data transmission periods in the context of the proposed SeCE and PS techniques.

\subsection{Channel estimation period}
%Motivated by the large number of ports in the small space of a UE's FA, and therefore the existence of a huge mutual coupling between the FA's ports, we propose a 
In the context of the considered limited coherence interval scenario, we propose a novel LMMSE-based channel estimation technique, namely SeCE technique, where the channel estimation between the UEs' ports and their serving BSs is performed via pilot-training symbols in a sequential manner. The pilot-training symbols, which are initially known at both the BS and the UEs in the same cell, are broadcasted by each BS to its connecting UEs\footnote{Different mutually orthogonal pilot sequences are used by the BSs in all cells, and thus pilot contamination is neglected.}. Contrary to the conventional approaches, where channel estimation is performed for all FAs' ports, the proposed SeCE technique requires a channel estimation process for only a subset of ``selected'' ports for each FA. The proposed technique is motivated by the fact that, in practice, the number of ports, $N$, can be very large, and the estimation of $\left|g_{0i}^{(k)}\right|$ for all FAs' ports is infeasible under the considered limited coherence interval scenario. Fig. \ref{Prop} demonstrates our proposed SeCE technique which is sequentially employed by all FAs of the UEs. Initially, the serving BS  broadcasts pilot-training symbols to the reference (i.e., first) port, and their link's channel coefficients are estimated by employing a LMMSE channel estimation technique. Then, by exploiting the strong spatial correlation between adjacent FA's ports, the channel estimation process related to  the subsequent $\nu$ ports is skipped, by assuming that the channels for these ports are equal to the estimated channel of the previously ``selected'' port. Hence, the proposed SeCE technique proceeds with the channel estimation process of the link between the serving BS and the $(\nu+2)$-th port. The previous process is repeated until channel estimation for all links between the serving BS and the FA's ports is obtained. For each FA, the set defined by the ports for which the serving BS performs LMMSE channel estimation, represents the set of ``selected'' FA's ports i.e., $\mathfrak{N}\subseteq\mathcal{N}$. Therefore, as illustrated in Fig. \ref{SeCE}, our proposed technique requires channel estimation for only $N'=\lceil\frac{N}{\nu+1}\rceil$ links for each FA, opposed to the conventional techniques (illustrated in Fig. \ref{Conv}) that require channel estimation for $N$ links. The reduction of the required channel estimation coefficients becomes even more apparent in the considered multi-FA environments, resulting in a significantly mitigated signaling overhead among the BSs and the UEs. On the other hand, the reduced channel resolution caused by the sparse channel estimation process, jeopardizes the observed SINR. Hence, the employment of our proposed SeCE technique triggers a non-trivial trade-off between mitigating the signaling overhead related to the channel estimation process and reducing the observed SINR. It is important to mention here that, by considering $\nu=0$, our proposed channel estimation scheme becomes the conventional LMMSE channel estimation scheme, where channel estimation is performed for all FAs' ports, and referred as \textit{Skip-less CE technique}.
\begin{figure}
	\centering
	\begin{subfigure}{.5\textwidth}
		\centering
		\includegraphics[width=.8\linewidth]{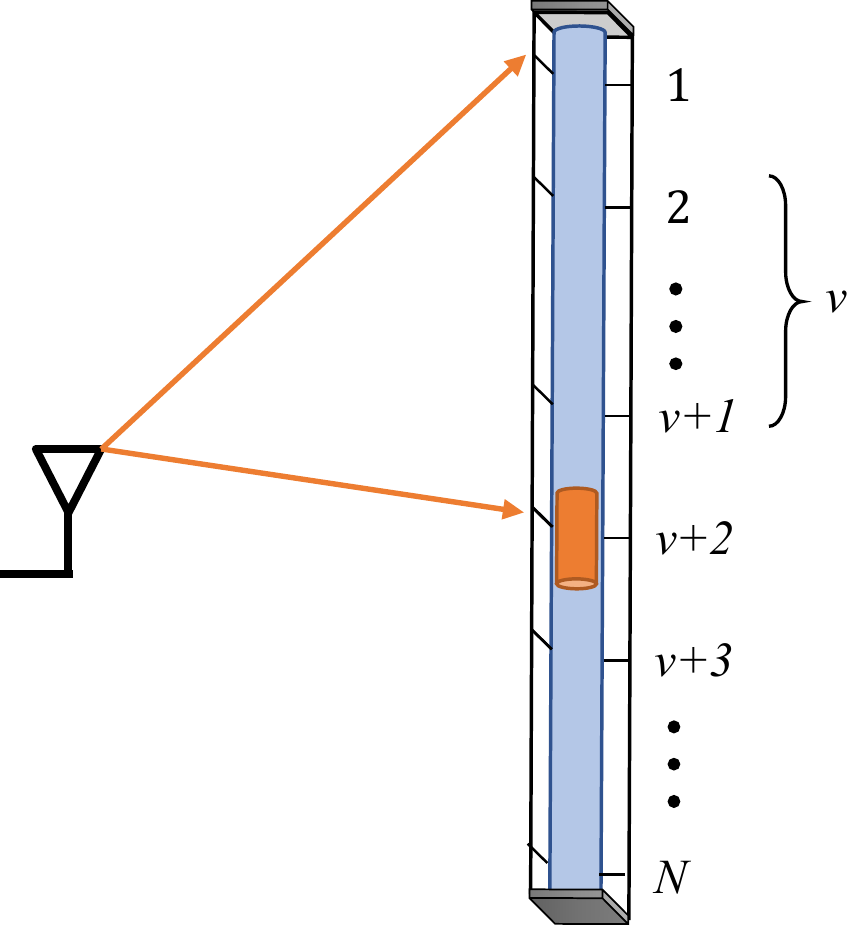}
		\caption{The proposed SeCE technique.}\label{Prop}
	\end{subfigure}%
	\begin{subfigure}{.5\textwidth}
		\centering
		\includegraphics[width=.69\linewidth]{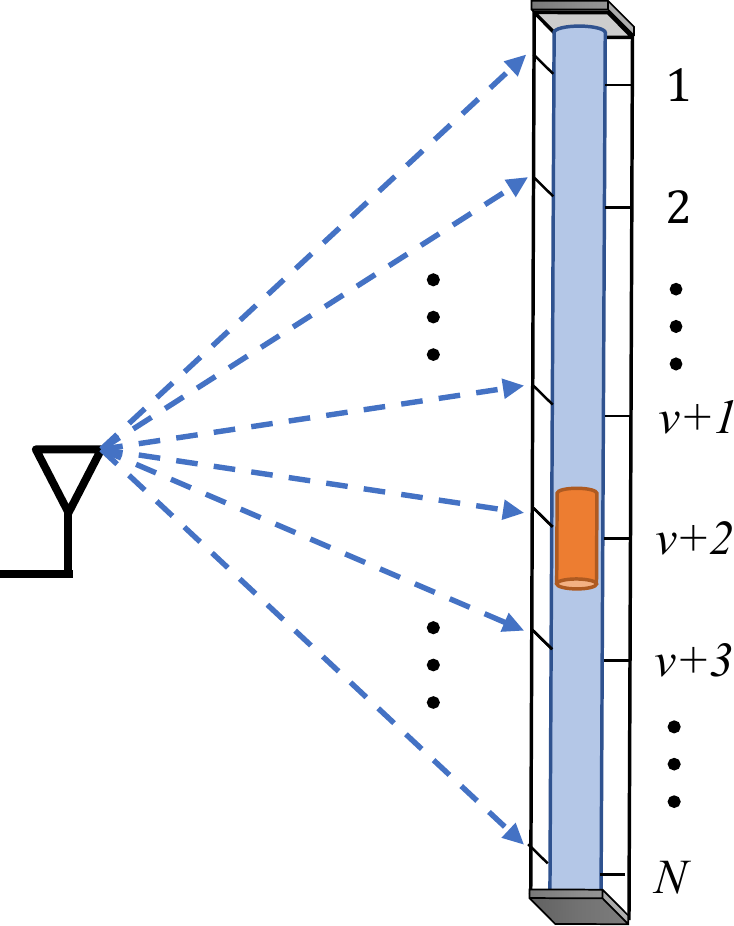}
		\caption{The conventional Skip-less CE technique.}\label{Conv}
	\end{subfigure}
	\caption{Comparison of Skip-less CE and SeCE techniques.}
	\label{SeCE}\vspace{-20pt}
\end{figure}

By taking into account the FA model in Section \ref{FAmodel}, a portion of the channel uses allocated for the first phase (i.e., channel estimation period) is devoted to the flow motion of the fluid metal between the various ports of a FA, also known as \textit{switching channel uses} and denoted as $l_s$. This reflects the fact that a fluid metal moves at a finite velocity within the electrolyte-filled capillary of a FA, in contrast to existing works that assume high velocity fluid metals that leads to instantaneous port switching i.e., $l_s=0$. In the context of our proposed SeCE technique, the fluid metal of each FA performs $N'-1$ transitions of duration $\delta(\nu+1)$ between adjacent ports that belong in the set of ``selected'' FA's ports (see Fig. \ref{Prop}). Hence, in the context of the adopted circular $M$-FA array, the required switching channel uses are equal to $l_s=M(N'-1)\delta(\nu+1)W_c$. The key idea of the adopted channel estimation process is to divide the remaining channel estimation period, i.e. $L_e-l_s$, into $N'M$ symmetric segments of $\Delta=\frac{L_e-l_s}{N'M}$ consecutive symbols (see Fig. \ref{Time}). During each segment, the channel between a single FA port and the serving BS is estimated\footnote{More sophisticated vector-based channel estimation schemes can be used but this estimation process is sufficient for the purpose of this work.}. Therefore, the baseband equivalent received pilot signal at the $i$-th port of the $k$-th FA of the typical UE, is given by
\begin{equation}
\mathbf{y}_i^{(k)} = \sqrt{\frac{\Delta P}{\ell(r_i(\rho))}}g_{{\rm 0}i}^{(k)} \mathbf{X}_0^{(k)}+\sum_{\substack{j\in\mathbb{N}^+ \\ x_j\in\Phi\backslash x_{\rm 0}}}\sqrt{\frac{P}{\ell\left(\|x_j\|\right)}}g_{ji}^{(k)}\mathbf{X}_j+\eta_0,
\end{equation}
where $\mathbf{X}_0$ is a deterministic $\Delta\times 1$ training symbol vector satisfying $\mathbf{X}_0^\dagger \mathbf{X}_0=1$ \cite{KAY}, $\mathbf{X}_k\stackrel{d}{\sim}\mathcal{CN}\left(\mathbf{0}_{\Delta\times 1},\mathbf{I}_{\Delta}\right)$ is a transmission symbol vector from node $k$, $g_{ji}^{(k)}\stackrel{d}{\sim}\mathcal{CN}(0,1)$ depicts the channel between the typical UE and the $k$-th interfering BS, $\eta_0\stackrel{d}{\sim}\mathcal{CN}\left(\mathbf{0}_{\Delta\times 1},N_0\mathbf{I}_{\Delta}\right)$ is the additive white Gaussian noise (AWGN) vector, and  $X\stackrel{d}{\sim}Y$ implies that X \textit{is distributed as} Y. Note that $g_{{\rm 0}i}^{(k)}$, $g_{ji}^{(k)}$, $\mathbf{X}_0$, $\mathbf{X}_j$, and $\eta_0$ are independent between each other. Hence, the observation scalar signal $y_i^{(k)}=\mathbf{X}_0^\dagger\mathbf{y}_i^{(k)}$ at the $i$-th port of the $k$-th FA of the typical UE can be evaluated as
\begin{equation}
	y_i^{(k)}=\sqrt{\frac{\Delta P}{\ell(r_i(\rho))}}g_{{\rm 0}i}^{(k)}+\sum_{\substack{j\in\mathbb{N}^+ \\ x_j\in\Phi\backslash x_{\rm 0}}}\sqrt{\frac{P}{\ell\left(\|x_j\|\right)}}g_{ji}^{(k)}q_j+z_0,
\end{equation}
where $q_j=\mathbf{X}_0^\dagger\mathbf{X}_j\stackrel{d}{\sim}\mathcal{CN}(0,1)$ and $z_0=\mathbf{X}_0^\dagger \eta_0\stackrel{d}{\sim}\mathcal{CN}(0,N_0)$.

By using the low-complexity LMMSE estimator, which is optimal among the class of linear estimators, the estimate of $g_{{\rm 0}i}^{(k)}$ conditioned on $\rho$, is given by 
\begin{equation}\label{VarianceProof}
\widehat{g}_{{\rm 0}i}^{(k)}|_{\rho}=\frac{\sqrt{\frac{\Delta P}{\ell(r_i(\rho))}}\sigma^2}{\frac{\Delta P}{\ell(r_i(\rho))}\sigma^2+N_0+P\mathbb{E}(\mathcal{I})}y_i^{(k)},
\end{equation}
where $\mathcal{I}$ denotes the instantaneous power of the interference at the typical UE, and is given by 
\begin{equation}\label{Interference}
	\mathcal{I}=\sum_{\substack{j\in\mathbb{N}^+ \\ x_j\in\Phi\backslash x_{\rm 0}}}\frac{1}{\ell\left(\|x_j\|\right)}\left| g_{ji}^{(k)}\right|^2 .
\end{equation}
Based on the Campbell's Theorem \cite{SGbook} and conditioning on $\rho$, $\mathbb{E}(\mathcal{I})$ can be expressed as
\begin{equation}\label{MeanInterference}
	\mathbb{E}(\mathcal{I}) = 2\pi\lambda_b\mathbb{E}\left[\left| g_{ji}^{(k)}\right|^2 \right]\int_{r_i(\rho)}^{\infty}r^{1-a}{\rm d}r=2\pi\lambda_b\sigma^2\frac{r_i(\rho)^{2-a}}{a-2},
\end{equation}
where $r_i(\rho)$ depicts the distance of the typical link between the $i$-th port of a FA and the tagged BS, which is given by \eqref{Distance}. 

The channel estimation error can then be derived as $e_i|_{\rho} = g_{{\rm 0}i}^{(k)}-\widehat{g}_{{\rm 0}i}^{(k)}|_{\rho}$, where $e_i|_{\rho}\stackrel{d}{\sim}\mathcal{CN}\left(0,\sigma^2_{e}|_{\rho}\right)$, and $\widehat{g}_{{\rm 0}i}^{(k)}|_{\rho}$ and $e_i|_{\rho}$ are uncorrelated \cite{KAY}. An explicit expression for the variance of the channel estimation conditioned on $\rho$, $\sigma^2_{e}|_{\rho}= \mathbb{E}\left[\big(g_i-\hat{g}_i|_{\rho}\big)^2\right]$, is given in Section \ref{PerformanceAnalysis}.

\subsection{Data transmission period}
Regarding the data transmission period (i.e., for the rest of the coherence block after the channel estimation process), we propose a low-complexity PS technique which is based on a two-stage procedure. In the first stage, the set of $M$ ``candidate'' ports for each UE is formed, that consists of the pre-selected ports with which the serving BS is able to communicate. More specifically, a single port is chosen among the set of ``selected'' ports from each FA, i.e. $i\in\mathfrak{N}$, that is estimated to provide the strongest channel in order to have the best reception performance. Hence, the $k$-th FA's location of the typical UE is switched to the port that satisfies 
\begin{equation}\label{Association}
	\mathtt{i}_k=\argmax_{i\in\mathfrak{N}}\left\lbrace \Big|\widehat{g}_{{\rm 0}i}^{(k)}\Big|\right\rbrace, \quad\forall k\in\mathcal{M}.
\end{equation}
The set defined by the pre-selected ports from all FAs, represents the set of ``candidate'' ports i.e., $\mathcal{C}=\{\mathtt{i}_{k}:\ k\in\mathcal{M}\}$. At the second stage, each UE selects from the $M$ ``candidate'' ports, the one which provides the highest SINR, while the rest of the ``candidate'' ports are ignored. The low complexity of the proposed technique stems from the fact that a single RF chain is required since the UEs connect to their serving BSs via a single port, opposed to the conventional maximum ratio combining (MRC) approach in which a UE communicates with its serving BS with $M$ ports, demanding the existence of $M$ RF chains \cite{NIG}. This is where the novelty of our proposed technique is highlighted. Particularly, by applying the proposed technique, the complexity of implementing port diversity schemes in multi-FA communication networks can be significantly decreased.

For simplicity, we consider that both the data and the pilot-training symbols are transmitted with the same power $P$ (dBm). Thus, the received signal at the $\mathtt{i}_k$-th pre-selected port of the $k$-th FA of the typical UE during the $n$-th channel use, is given by
\begin{equation}\label{DataTransmission}
	d_{\mathtt{i}_k}[n] = \sqrt{\frac{1}{r_{\mathtt{i}_k}^a(\rho)}}\widehat{g}_{{\rm 0}\mathtt{i}_k}^{(k)} s_{\rm 0}[n]+\sqrt{\frac{1}{r_{\mathtt{i}_k}^a(\rho)}}e_{\mathtt{i}_k} s_{\rm 0}[n]+\sum_{\substack{j\in\mathbb{N}^+ \\ x_j\in\Phi\backslash x_{\rm 0}}}\sqrt{\frac{1}{\|x_j\|^a}}g_{j\mathtt{i}_k}^{(k)}s_j[n]+\eta_0[n],
\end{equation}
where $n\in\{L_e+1,\dots, L_c\}$, $s_0[n]$ and $s_k[n]$ represent independent Gaussian distributed data symbols from the tagged and the $k$-th interfering BS, respectively, satisfying $\mathbb{E}\left[|s_0[n]|^2\right]=P$ and $\mathbb{E}\left[|s_j[n]|^2\right]=P$; $\eta_0[n]\stackrel{d}{\sim}\mathcal{CN}(0,N_0)$ is AWGN. Note that, the first term of \eqref{DataTransmission} is known at the receiver, while the remaining terms are unknown and are treated as noise. Therefore, an estimate of $s_0[n]$ can be formulated as $\hat{s}_0[n]=\sqrt{r_{\mathtt{i}_k}^a(\rho)}\frac{\left(\widehat{g}_{{\rm 0}\mathtt{i}_k}^{(k)}\right)^*}{\Big|\widehat{g}_{{\rm 0}\mathtt{i}_k}^{(k)}\Big|^2}d_{\mathtt{i}_k}[n]$, from which the SINR observed at the $\mathtt{i}_k$-th port can be written as
\begin{equation}\label{SINR}
	{\rm SINR}_{\mathtt{i}_k} = \frac{\frac{\epsilon}{r_{\mathtt{i}_k}^a(\rho)}\left| \widehat{g}_{{\rm 0}\mathtt{i}_k}^{(k)}\right|^2}{\sum\limits_{\substack{j\in\mathbb{N}^+ \\ x_j\in\Phi\backslash x_0}}\frac{\epsilon}{\|x_j\|^a}\left| g_{ji}^{(k)}\right|^2+\frac{\epsilon}{r_{\mathtt{i}_k}^a(\rho)}\sigma^2_e|_{\rho}+1},
\end{equation}
where $\epsilon=\frac{\sigma^2P}{N_0}$ is the transmit signal-to-noise ratio (SNR).

\section{Outage Performance for SeCE and PS Techniques}\label{PerformanceAnalysis}
In this section, we analytically evaluate the outage performance of a DL homogeneous cellular network in the context of the proposed SeCE and PS techniques. Initially, we evaluate the variance of the channel estimation under the proposed SeCE technique and study the statistical properties of both the estimated channels under the LMMSE estimator and the induced multi-user interference. Finally, the outage performance of the considered system model is evaluated, by leveraging tools from SG.

\subsection{Preliminary results}
In this section we state some preliminary results, which will assist in the derivation of the main analytical framework. To begin with, the variance of the channel estimation under the proposed SeCE technique, conditioned on $\rho$, is evaluated in the following proposition.
\begin{proposition}\label{Proposition1}
	By employing the SeCE technique, the variance of the channel estimation error observed at the $i$-th port of the typical UE's $k$-th FA, conditioned on $\rho$, is given by 
	\begin{equation}\label{NoiseVarianceEq}
	\sigma^2_{e}|_{\rho} =\left(1+\frac{\Delta}{\frac{r_i^a(\rho)}{\epsilon}+2\pi\lambda_b\frac{r_i^{2}(\rho)}{a-2}}\right)^{-1}, 
	\end{equation} 
	where $\Delta=\frac{L_e-l_s}{N'M}$ and $l_s=M(N'-1)\delta(\nu+1)W_c$.
\end{proposition}
\begin{proof}
	See Appendix \ref{Appendix1}.
\end{proof}
Proposition \ref{Proposition1} clearly demonstrates that the variance of channel estimation error is a non-negative increasing concave function with respect to both the number of the FAs $M$ and the number of ports composed in each FA $N$. For the extreme case of infinite number of FAs ports i.e., $M\rightarrow\infty$ and/or $N\rightarrow\infty$, the variance of the channel estimation error (i.e., channel estimation quality) becomes equal to one i.e., $\sigma^2_{e}|_{\rho}\approxeq 1$. As stated in \eqref{SINR}, this inadequate channel estimation quality induces a substantial reduction of the SINR observed by the FA-based UEs, compromising the network performance. Consequently, although the increased number of FA ports initially enhances the receive diversity gain and thereby the network performance, beyond a critical point $N^*\in\mathcal{N}$, a further increase of the number of FA ports leads to an attenuated channel estimation quality, jeopardizing the overall network performance. Based on the aforementioned discussion, the number of FA ports triggers a trade-off between improving the network's outage performance and reducing the channel estimation quality.

By leveraging the flexibility offered by the proposed SeCE scheme in selecting the value of $\nu$, in the following proposition we provide a closed-form expression for the minimum number of skipped ports $\nu^*$ in such a way as to achieve a pre-defined desired variance of CE error.
\begin{proposition}\label{Proposition2}
The minimum number $\nu^*$, which ensures a variance of channel estimation error $\varsigma\in(0,1)$, is equal to
\begin{equation}
\nu^*=\left(\frac{\frac{r_i^a(\rho)}{\epsilon}+\mathbb{E}[\mathcal{I}]}{\frac{L_e}{MN}-\frac{\kappa\lambda W_c}{(N-1)u}}\right)\left(\varsigma^{-1}-1\right),
\end{equation}
where $u$ depicts the average velocity achieved by a fluid metal, and is given by \eqref{velocity}.
\end{proposition}
\begin{proof}
	The expression is derived by solving the expression derived in Proposition \ref{Proposition1} with respect to the number of skipped FA ports, and by using the inequality $N'=\lceil\frac{N}{\nu+1}\rceil\geq\frac{N}{\nu+1}$.
\end{proof}
Based on the Proposition \ref{Proposition2}, it can be observed that the minimum number of ports that need to be skipped is a non-negative decreasing convex function with respect to the targeted channel estimation error variance $\varsigma$. Therefore, the demand for higher channel estimation quality i.e., $\varsigma\rightarrow 0$, requires the allocation of more channel uses, and thus pilot-training symbols, for all ``selected'' ports of the UE's FAs. In the context of the considered limited coherence interval scenario, this can be achieved with the larger number of skipped FA's ports, e.g. $ \nu^* \rightarrow N-1$, and hence less channel estimation processes. Contrary, as the required channel estimation quality decreases i.e., $\varsigma\rightarrow 1$, the minimum number of skipped FA's ports also decreases $\nu^*\rightarrow 0$.

Regarding the statistical properties of the estimated channels under the proposed SeCE technique, the joint pdf and cumulative distribution function (cdf) of $|\hat{g}_1|,\dots,|\hat{g}_N|$, conditioned on $\rho$, are given in the following lemma.
\begin{lemma}\label{Lemma1}
	The conditional joint cdf and pdf of $\left| \widehat{g}_{{\rm 0}1}^{(k)}\right|,\dots,\left|\widehat{g}_{{\rm 0}N'}^{(k)}\right|$ for $i\in\mathfrak{N}$ and $k\in\mathcal{K}$, are given by
	\begin{equation}\label{CDF}
		F_{\left| \widehat{g}_{{\rm 0}1}^{(k)}\right|,\dots,\left|\widehat{g}_{{\rm 0}N'}^{(k)}\right|}(\tau_1,\dots,\tau_{N'}|\rho)=\int_0^{\frac{\tau_1^2}{\widetilde{\sigma}_1^2}}\exp\left(-t\right)\prod_{j\in\mathfrak{N}}\left[1-Q_1\left(\sqrt{2\mu_j^2\frac{\widetilde{\sigma}_1^2}{\widetilde{\sigma}_j^2}t},\sqrt{\frac{2}{\widetilde{\sigma}_j^2}}\tau_j\right)\right]{\rm d}t
	\end{equation}
and
\begin{equation}\label{PDF}
	f_{\left| \widehat{g}_{{\rm 0}1}^{(k)}\right|,\dots,\left|\widehat{g}_{{\rm 0}N'}^{(k)}\right|}(\tau_1,\dots,\tau_{N'}|\rho)=\prod_{\substack{i\in\mathfrak{N}\\ (\mu_1\triangleq 0)}}\frac{2\tau_i}{\widetilde{\sigma}_i^2}\exp\left(-\frac{\tau_i^2+\mu^2_i\tau_1^2}{\widetilde{\sigma}_i^2}\right)I_0\left(\frac{2\mu_i\tau_1\tau_i}{\widetilde{\sigma}_i^2}\right),
\end{equation}
respectively, where $\tau_1,\dots,\tau_{N'}\geq 0$, $I_0(\cdot)$ represents the zero-order modified Bessel function of the first kind, $Q_1(\cdot,\cdot)$ is the first-order Marcum Q-function, and $\widetilde{\sigma}_i^2=\sigma^2(1-\mu_i^2)+\sigma^2_{e}|_{\rho}$.
\end{lemma}
\begin{proof}
	See Appendix \ref{Appendix2}.
\end{proof}
The performance of a FA-based UE in large-scale multi-cell networks is mainly compromised by the existence of multi-user interference \cite{SGbook}. In the considered network deployment, the multi-user interference observed by the $i$-th port of the typical UE's $k$-th FA, where the serving BS is located at $x_0\in\Phi$, is given by \eqref{Interference}.  Although the performance of a communication network by considering the actual multi-user interference can be easily evaluated for the PPP case with independent fading channels, in most relevant (realistic) models, it is either impossible to analytically analyze or cumbersome to evaluate even numerically. Motivated by the aforementioned discussion, the following proposition states our assumption of approximate the multi-user interference distribution of large-scale wireless networks by using Gamma distribution, aiming to provide simple and tractable expressions for the outage performance.
\begin{proposition}\label{InterferenceDistribution}
The multi-user interference observed at a port of the typical UE, conditioned on the distance from the serving BS, follows a Gamma distribution with pdf
\begin{equation}\label{Gammapdf}
f_{\mathcal{I}}(\gamma|r)=\frac{\gamma^{\varpi-1}\exp\left(-\frac{\gamma}{\varrho}\right)}{\Gamma[\varpi]\varrho^\varpi},\ \gamma>0,
\end{equation}
with shape parameter $\varpi=2\left(\pi\lambda_b\frac{r^{2-a}}{a-2}\right)^2$ and scale parameter $\varrho=\frac{\sigma^2(a-2)}{\pi\lambda_br^{2-a}}$.
\end{proposition}
\begin{proof}
	See Appendix \ref{Appendix3}.
\end{proof}
\subsection{Outage performance}\label{Coverage}
By using the proposed PS technique, the serving BS selects the port from the set of ``candidate'' ports, $\mathcal{C}$, that provides the maximum SINR i.e. ${\rm SINR}^*=\max\{{\rm SINR}_{\mathtt{i}_k}\}$, where $k\in\mathcal{M}$. Hence, the downlink outage performance $\mathcal{P}_o(R)$ of a network can be described as the probability that the maximum mutual information of the channel between the typical UE and its serving BS is smaller than a target rate $R$ (data bits/channel use). This performance can be mathematically described with the probability $\mathbb{P}\left[\left(1-\frac{L_t}{L_c}\right)\log\left(1+{\rm SINR}^*\right)<R\right]$, where $\left(1-\frac{L_t}{L_c}\right)$ represents the fractional amount of time (relative to the total frame length) used for data transmission. Therefore, with uncorrelated branches, the conditional cdf of ${\rm SINR}^*$ is given by
\begin{align}
	\mathcal{P}_o(R)&=\mathbb{P}\left[{\rm SINR}^*<2^{\frac{R}{1-\frac{L_t}{L_c}}}-1\right]\nonumber\\&=\prod_{k=1}^M\mathbb{E}_{\rho,\mathcal{I}_{\mathtt{i}_k}}\left[\mathbb{P}[{\rm SINR}_{\mathtt{i}_k}<\vartheta|\rho,\mathcal{I}_{\mathtt{i}_k}]\right]\label{ProofOutage1},
\end{align}
where $\mathbb{P}[{\rm SINR}_{\mathtt{i}_k}<\vartheta|\rho,\mathcal{I}_{\mathtt{i}_k}]$ denotes the outage probability of the typical UE communicating with the serving BS via the pre-selected port at its $k$-th FA. By substituting \eqref{SINR} into \eqref{ProofOutage1}, the conditional cdf of ${\rm SINR}^*$ can be calculated as
\begin{align}
	\mathcal{P}_o(R)&=\prod_{k=1}^M\mathbb{E}_{\rho,\mathcal{I}_{\mathtt{i}_k}}\left[\mathbb{P}\left[\left|\widehat{g}_{{\rm 0}\mathtt{i}_k}^{(k)}\right|<\sqrt{\vartheta r_{\mathtt{i}_k}^a(\rho)\left(\mathcal{I}_{\mathtt{i}_k}+\frac{\sigma^2_e|_{\rho}}{r^{a}_{\mathtt{i}_k}(\rho)}+\frac{1}{\epsilon}\right)}\middle|\rho,\mathcal{I}_{\mathtt{i}_k}\right]\right]\nonumber\\
	&=\prod_{k=1}^M\mathbb{E}_{\rho,\mathcal{I}_{\mathtt{i}_k}}\left[\mathbb{P}\left[\left|\widehat{g}_{{\rm 0}\mathtt{i}_k}^{(k)}\right|<\sqrt{\Theta_{\mathtt{i}_k}}\middle|\rho,\mathcal{I}_{\mathtt{i}_k}\right]\right]\label{Proof},
\end{align}
where $\vartheta=2^{\frac{R}{1-\frac{L_t}{L_c}}}-1$ and $\Theta_i = \vartheta r_{i}^a(\rho)\left(\mathcal{I}_{i}+\frac{\sigma^2_e|_{\rho}}{r^{a}_{i}(\rho)}+\frac{1}{\epsilon}\right)$. In the following lemma, an analytical expression for the conditional outage performance i.e., $\mathcal{P}_o(R|\rho,\mathcal{I}_{\mathtt{i}_k})=\mathbb{P}\left[\left|\widehat{g}_{{\rm 0}\mathtt{i}_k}^{(k)}\right|<\sqrt{\Theta_{\mathtt{i}_k}}|\rho,\mathcal{I}_{\mathtt{i}_k}\right]$ with $k\in\mathcal{M}$, of the considered system model is derived.

\begin{lemma}\label{Lemma2}
	The conditional outage probability when utilizing the proposed SeCE  technique in cellular networks with FA-based UEs, is given by
	\begin{align}\label{Conditional}
\mathcal{P}_o(R|\rho,\mathcal{I}_{\mathtt{i}_k}) = \int_0^{\frac{\Theta_1^2}{\widetilde{\sigma}_1^2}}\exp\left(-t\right)\prod_{j\in\mathfrak{N}}\left[1-Q_1\left(\sqrt{2\mu_j^2\frac{\widetilde{\sigma}_1^2}{\widetilde{\sigma}_j^2}t},\sqrt{\frac{2}{\widetilde{\sigma}_j^2}}\Theta_j\right)\right]{\rm d}t,
	\end{align}
where $\Theta_j = \vartheta r_j^a(\rho)\left(\mathcal{I}_j+\frac{\sigma^2_e|_\rho}{r^{a}_j(\rho)}+\frac{1}{\nu}\right)$; $r_j(\rho)$ and $\sigma^2_{e}|_{\rho}$ depict the length of the typical link and the channel estimation, that are given by \eqref{Distance} and \eqref{NoiseVarianceEq}, respectively, and $\widetilde{\sigma}_j^2=\sigma^2(1-\mu_j^2)+\sigma^2_{e}|_{\rho}$.
\end{lemma}
\begin{proof}
	The conditional outage probability expression can be calculated as
	\begin{align*}
		\mathcal{P}_o(R|\rho,\mathcal{I}_{\mathtt{i}_k})&=\mathbb{P}\left[\left|\widehat{g}_{{\rm 0}\mathtt{i}_k}^{(k)}\right|<\sqrt{\Theta_{\mathtt{i}_k}}\middle|\rho,\mathcal{I}_{\mathtt{i}_k}\right]\nonumber\\
		&=\mathbb{P}\left[\left|\widehat{g}_{{\rm 0}1}^{(k)}\right|<\sqrt{\Theta_1},\dots,\left|\widehat{g}_{{\rm 0}N'}^{(k)}\right|<\sqrt{\Theta_{N'}}|\rho,\mathcal{I}_1,\dots,\mathcal{I}_{N'}\right].
	\end{align*}
	Thus, by substituting $\tau_1=\sqrt{\Theta_1}$,$\dots$, $\tau_N=\sqrt{\Theta_N}$ into the joint cdf that is given in Lemma \ref{Lemma1}, which completes the proof.
	%Based on the adopted association scheme, that is given by \eqref{Association}, and the distribution of the multi-user interference experienced by the $i$-th port of the typical UE (see Proposition \ref{InterferenceDistribution}), the outage performance conditioned on $\{r_i\}$, where $i\in\{1,N\}$, can be evaluated as
	%{\small \begin{align}
	%	&\mathcal{P}_c(\vartheta|\{r_i\})\!=\!\prod_{i=1}^N\int_{\gamma_i=0}^\infty\nonumber\mathbb{P}\!\left[|\hat{g}_1|\!<\!\sqrt{\vartheta r_\mathtt{i}^a\!\left(\mathcal{I}_\mathtt{i}\!+\!\frac{\sigma^2_e|_{r_\mathtt{i}}}{r^{a}_\mathtt{i}}\!+\!\frac{1}{\nu}\right)},\dots,|\hat{g}_N|\!<\!\sqrt{\vartheta r_\mathtt{i}^a\!\left(\mathcal{I}_\mathtt{i}\!+\!\frac{\sigma^2_e|_{r_\mathtt{i}}}{r^{a}_\mathtt{i}}\!+\!\frac{1}{\nu}\right)}|\{r_i\},\{\gamma_i\}\right]\\&\qquad\qquad\qquad\qquad\qquad\qquad\qquad\qquad\qquad\qquad\qquad\qquad\qquad\qquad\qquad\qquad\qquad\qquad\times f_{\mathcal{I}}(\gamma_i;k,\theta){\rm d}\gamma_i,
	%\end{align}}where the probability inside the integral can be evaluated as \eqref{CDF}. By unconditioning on $\{r_i\}$ with the aid of \eqref{ContactPDF} and by letting $\delta_{j} = \frac{\vartheta r_j^a}{\widetilde{\sigma}_j^2}\left(\gamma_j+\frac{\sigma^2_e|_{r_j}}{r^{a}_j}+\frac{1}{\nu}\right)$, the final expression for the achieved outage probability can be derived.
\end{proof}
In spite of the fact that Lemma \ref{Lemma2} provides an analytical approach to obtain the conditional outage probability when utilizing the proposed SeCE technique, the analysis of the achieved performance is still cumbersome and tedious, impeding the extraction of meaningful insights. To this end, we evaluate the achieved performance in the asymptotic regime. More specifically, by considering the interference-limited scenario (i.e., $\epsilon\rightarrow\infty$) for the special case where $\mu_i=\mu\ \forall i\in\mathcal{N}$ in the following lemma, an upper and a lower bound for the conditional outage probability can be derived.
\begin{lemma}\label{Lemma3}
	The conditional outage probability when utilizing the proposed SeCE technique in cellular networks with FA-based UEs, is upper bounded by
	\begin{equation}
	\mathcal{P}^{ u}_o(R|\rho,\mathcal{I}_{\mathtt{i}_k})\! =\!  1\!-\exp\left(\!-\Xi_1\right)\!-\!\Upsilon^u(\Xi_1)\!\sum_{j\in\mathfrak{N}}\exp\left(-\Xi_j\right)\!-2\frac{\pi\mu^2}{(1+\mu^2)^{3/2}}\!\sum_{j\in\mathfrak{N}}\!\sqrt{\Xi_j}\exp\!\left(\!-\Xi_j\frac{2+\mu^2}{1+\mu^2}\right)
	\end{equation}
and lower bounded by
\begin{equation}
\mathcal{P}^{l}_o(R|\rho,\mathcal{I}_{\mathtt{i}_k}) =  1\!-\!\exp\left(\!-\Xi_1\right)\!-\!\Upsilon^l(\Xi_1)\sum_{j\in\mathfrak{N}}\exp\left(-\Xi_j\right)
\end{equation}
	where $\widetilde{\sigma}^2=\sigma^2(1-\mu^2)+2\pi\lambda_b\frac{N}{L_T}\frac{r(\rho)^{2}}{a-2}$, $\Xi_i = \frac{\Theta_i^2}{\widetilde{\sigma}^2}$ for $i\in\mathfrak{N}$, $\Upsilon^u(x) = \frac{1-\exp\left(-x\left(1-\mu\right)^2\right)}{1+\mu^2}$ and $\Upsilon^l(x) = \frac{1-\exp\left(-x\left(1+\mu\right)^2\right)}{1+\mu^2}$.
\end{lemma}

\begin{proof}
	See Appendix \ref{Appendix4}.
\end{proof}
In the following Theorem, we evaluate the outage performance of the considered system model, by un-conditioning the expression in Lemma \ref{Lemma2} with both the interference distribution and the pdf of the distance from the typical UE to its serving BS.
\begin{theorem}
The outage probability when utilizing the proposed SeCE and PS techniques in cellular networks with FA-based UEs, is given by
\begin{equation*}
\mathcal{P}_o(R)= \prod_{k=1}^M\prod_{i\in\mathfrak{N}}\int_{0}^\infty\left(\int_{0}^\infty \mathcal{P}_o(R|\rho,\gamma_i)f_\mathcal{I}(\gamma_i|\rho){\rm d}\gamma_i\right)  2\pi\lambda_b \rho\exp\left(-\pi\lambda_b  \rho^2\right){\rm d}\rho,
\end{equation*}
where $\mathcal{P}_o(R|\rho,\gamma_i)$ represents the achieved outage probability conditioned on the distance of the typical UE from its serving BS $\rho$ \eqref{Distance} and the observed multi-user interference $\gamma_i$, that is given in Lemma \ref{Lemma2}, $\varpi$ and $\varrho$ depict the shape and scale parameter of the multi-user interference distribution, respectively, where $\varpi=2\left(\pi\lambda_b\frac{r_i(\rho)^{2-a}}{a-2}\right)^2$ and $\varrho=\frac{\sigma^2(a-2)}{\pi\lambda_br_i(\rho)^{2-a}}$.
\end{theorem}
\begin{proof}
	The expression for the outage probability can be obtained by firstly substituting the expression derived in Lemma \ref{Lemma2} into \eqref{Proof}. Then, by conditioning the derived expression on $\rho$ and averaging out the aggregate multi-user interference using Proposition \ref{InterferenceDistribution}. Lastly, the final expression can be derived by getting rid of the condition on $\rho$ by utilizing \eqref{ContactPDF}.
\end{proof}

\section{Numerical Results}
\begin{table*}[t]\centering
	\caption{Simulation Parameters}\label{Table2}
	\scalebox{0.85}{
		\begin{tabular}{| l | l || l | l |}\hline
			\textbf{Parameter} 		& \textbf{Value} 									& \textbf{Parameter} 						& \textbf{Value}\\\hline
			BSs' density ($\lambda_b$) & $5\times10^{-5}$	& Dimension ratio ($D/L$) & $5$  \\\hline
			Number of FAs ($M$) & $4$  	&  Voltage difference ($\Delta\phi$) & $10$ V  \\\hline
			Number of ports ($N$) & $15$ 	&  Number of skipped ports ($\nu$) & $1$ V \\\hline
			Path-loss exponent ($a$) 	& $4$ 	&  Channel variance ($\sigma$) & $1$\\\hline
			Scaling constant ($\kappa$)	& $0.2$ 	&  Bandwidth ($W_c$) 	& $100$ MHz  \\\hline
			Wavelength ($\lambda$) 	& $6$ cm 	&  Coherence time ($T_c$)	& $50$ ms   \\\hline
			Initial charge ($q$) 	& $0.07$ V  &  Percentage of channel estimation period ($L_e/L_c\times 100\%$)	& $16\%$\\\hline
			Viscosity ($\mu$) 	& $0.002$  &  Noise Variance ($N_o$)	& $10^{-5}$\\\hline
	\end{tabular}}
\end{table*}
\begin{figure}
	\centering\includegraphics[width=.6\linewidth]{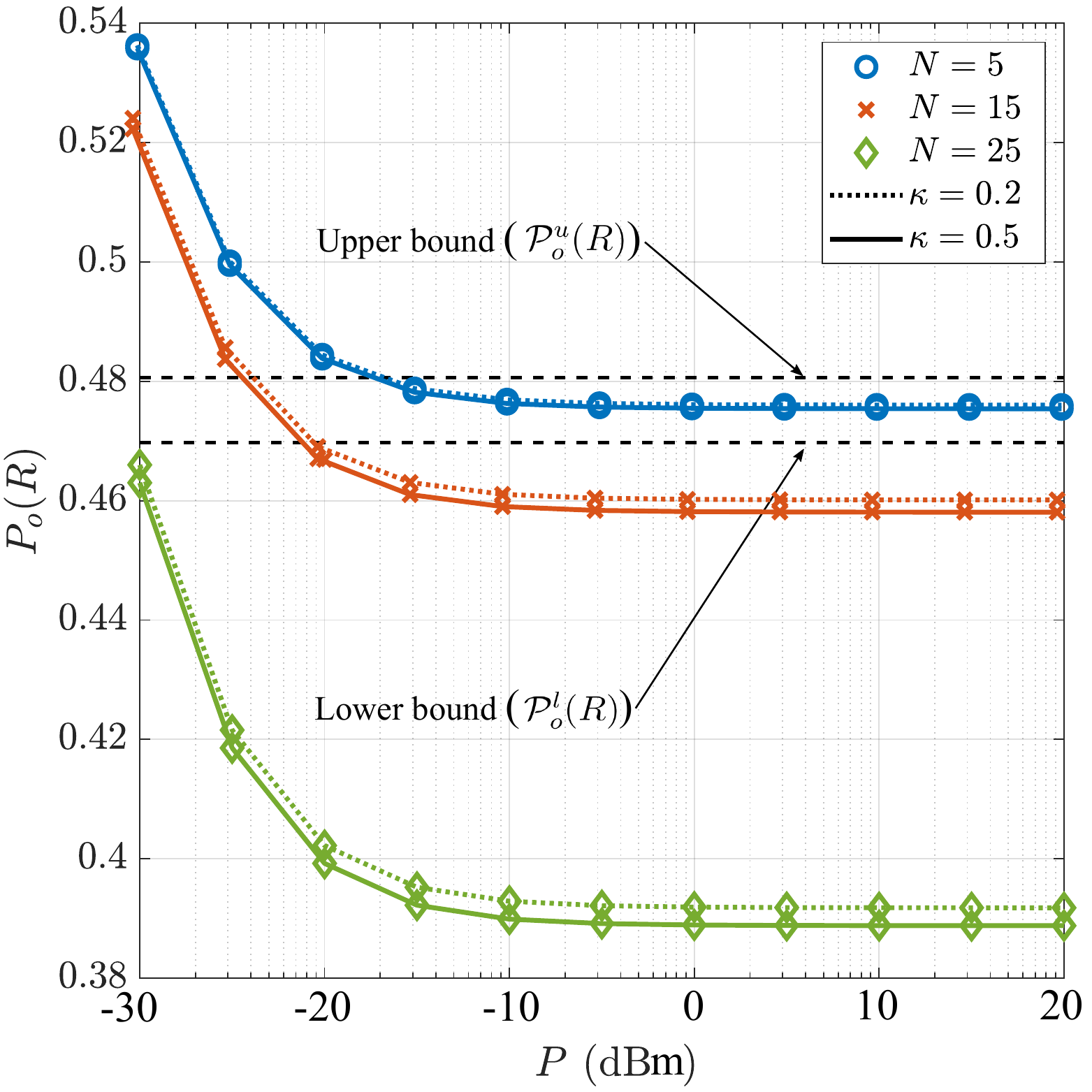}
	\caption{Outage performance versus the transmit power $P$ (dBm) for different $N$ and $\kappa$; $\nu=1$.}\label{FirstFigure}\vspace{-20pt}
\end{figure}
In this section, we provide numerical results to verify our model and illustrate the performance of FA-based UEs in large-scale cellular networks. A summary of the model parameters is provided in Table \ref{Table2}. Please note that, the selection of the simulation parameters is made for the sake of the presentation. The use of different values leads to a shifted network performance, but with the same observations and conclusions.

Fig. \ref{FirstFigure} reveals the effect of the transmit power on the outage performance achieved by a FA-enabled cellular network in the context of the proposed PS scheme. In particular, we plot the outage probability, $P_o(R)$, with respect to the transmit power $P$, for different scaling constants $\kappa=\{0.2,0.5\}$ and number of FAs' ports $N=\{5,15,25\}$. We can easily observe that larger FA architectures i.e., a larger $\kappa$, lead to a reduced outage probability. This was expected since, as the size of the FAs increases, the distance between their ports is also increases, limiting the negative effect of the spatial correlation between the ports' channels on the network performance. We can also observed that, the negative effect of the scaling constant on the achieved outage performance is minor for the considered sparse network deployment (i.e., $\lambda_b\rightarrow 0$). As expected, for sparse network deployments, the variance of the channel estimation error (i.e., channel estimation quality) becomes one i.e., $\sigma^2_{e}|_{\rho}\approxeq 1$, regardless of the scaling constant. This leads to a negligible effect of the scaling constant on the observed SINR and, consequently, on the achieved outage performance. Another important observation is that, the number of ports on each FA has a negative effect on the outage performance experienced by a UE in the considered network deployment. This observation can be explained by the fact that, a higher receive diversity gain can be achieved with the increased number of FA ports, and therefore, UEs observe an enhanced SINR. Additionally, it is clear from the figure that the outage probability asymptotically converges to a constant value which is tightly approximated by the proposed upper and lower bounds derived in Lemma \ref{Lemma3}. This behavior of the outage performance is based on the fact that as the transmission power of the nodes increases, the additive noise in the network becomes negligible. Finally, the agreement between the theoretical curves (solid and dashed lines) and the simulation results (markers) validates our analysis.
\begin{figure}
	\centering\includegraphics[width=.6\linewidth]{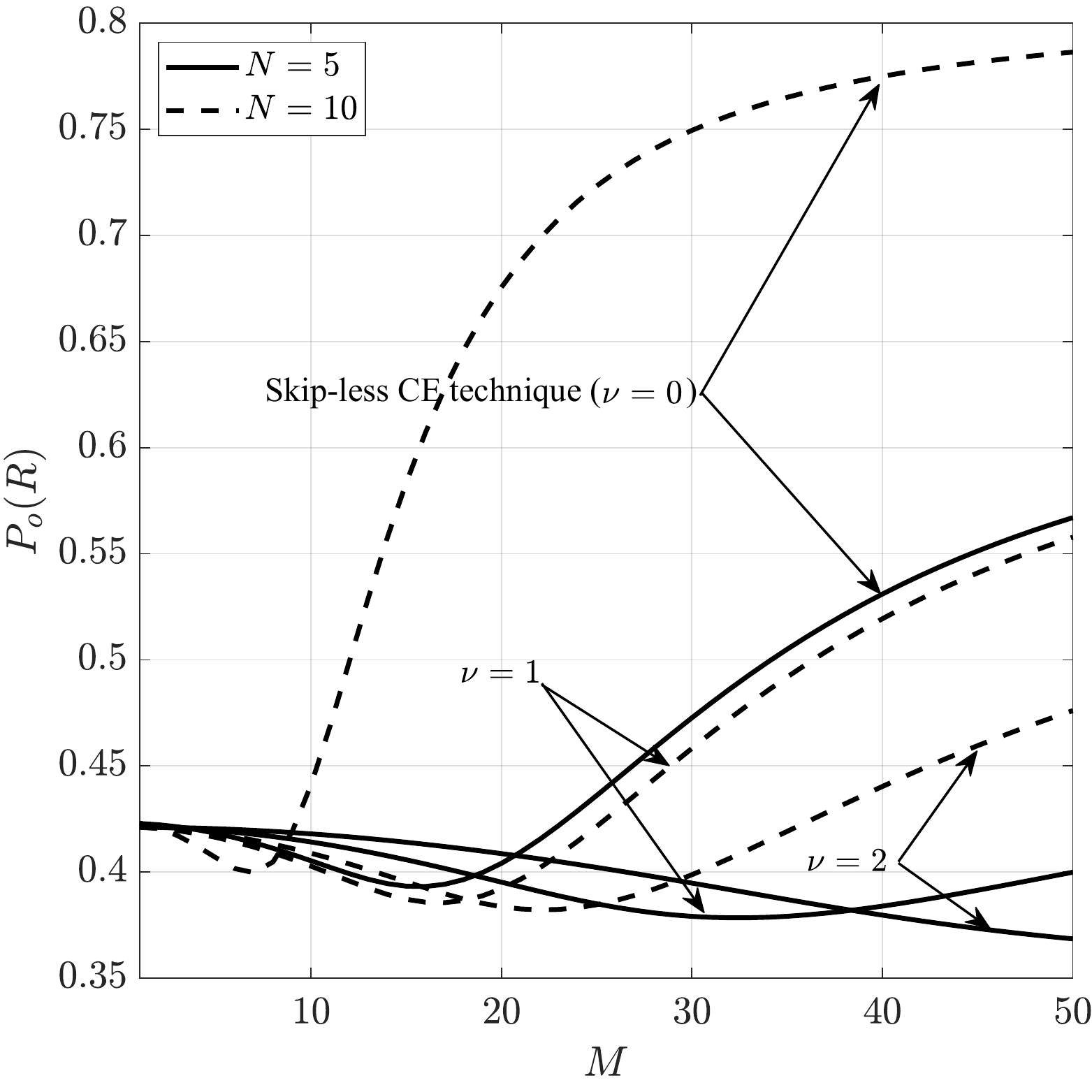}
	\caption{Outage performance versus the number of UEs' FAs $M$ for different $N$ and $\nu$.}\label{FirstFigureE}\vspace{-20pt}
\end{figure}

Fig. \ref{FirstFigureE} highlights the effectiveness of our proposed PS technique under the SeCE scheme compared to the conventional skip-less channel estimation technique for different number of FAs' ports $N=\{5,10\}$. More specifically, we plot the outage probability versus the  number of FAs that are equipped at each UE, $M$, for the proposed PS technique with $\nu=\{1,2\}$, as well as for the conventional skip-less channel estimation (denoted as ``Skip-less CE'') technique. An important observation is that, by increasing the number of FAs that are equipped at the UEs, the achieved network performance is initially enhances. This was expected since, the increased number of FAs, and consequently the elevated number of FAs' ports, results in a higher receive diversity gain, leading to a greater observed SINR at the UEs. Nevertheless, as the number of FAs at the UEs further increases, under the considered limited coherence interval scenario, the number of training-pilot symbols dedicated for the channel estimation of each port decreases, thereby the quality of the channel estimation is reduced (i.e., $\sigma^2_{e|\rho}\rightarrow 1$), jeopardizing the achieved network performance. An additional important observation which is derived from this figure is that, the conventional scheme provides a slightly better network performance compared to the proposed PS scheme for small number of FAs. This was expected, since with the utilization of the conventional scheme for a UE equipped with a low number of FAs, sufficient training-pilot signals are allocated to all FAs' ports, leading to an enhanced channel estimation quality and therefore an improved network performance. On the other hand, by utilizing the proposed PS scheme, the reduced channel estimation quality induced by neglecting the process of estimate the channel of some FAs' ports, results in a reduced network performance. However, by further increasing the number of FAs beyond a critical point, our proposed scheme overcomes the conventional scheme, providing a significantly enhanced network performance. As expected, our proposed technique reduces the number of channel estimation coefficients, enabling the channel estimation of a wider range of FAs' ports, thus, leading to an improved SINR observed by a UE and consequently an ameliorated network performance. Finally, we can observe that, by increasing the number of skipped ports, the optimal outage performance observed by FA-enabled UEs, which is achieved at a larger number of FAs, decreases. Therefore, by increasing the number of skipped ports, FA architectures with larger number of FAs can be supported that offered reduced outage performance, highlighting the practical interest of FA architectures with large number of FAs.
\begin{figure}
	\centering\includegraphics[width=.6\linewidth]{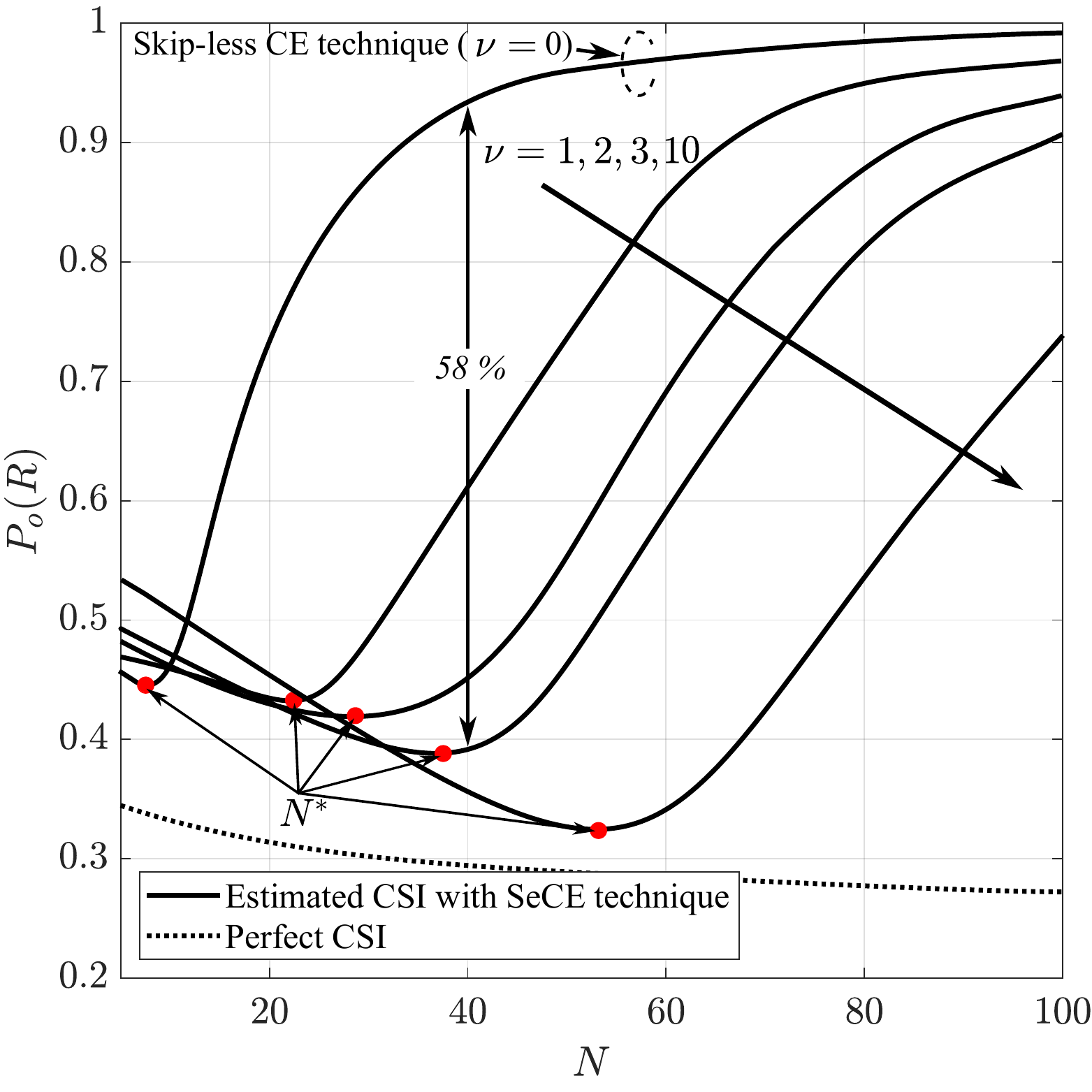}
	\caption{Outage performance versus the number of ports $N$ for different $\nu$.}\label{SecondFigure}\vspace{-20pt}
\end{figure}

Fig. \ref{SecondFigure} demonstrates the impact of the number of FAs' ports on the outage performance achieved by a FA-enabled communications in the context of the proposed PS scheme. More specifically, we plot the outage probability versus the number of FAs' ports, $N$, for different $\nu=\{1,2,3,10\}$. The first main observation is that for small number of FAs' ports, the presence of additional ports results in an elevation of the UEs' performance. On the other hand, we can easily observe that by increasing the number of FAs' ports beyond a critical point $N^*$, the outage performance increases. This is justified by the fact that, for a large number of FAs' ports, the allocated number of training-pilot symbols for the estimation of the channel of each port is decreases, thereby the quality of the channel estimation is diminished (i.e., $\sigma^2_{e|\rho}\rightarrow 1$), decreasing the achieved network performance. Regarding the effect of the proposed SeCE technique on the achieved outage performance, we can easily observe that the critical number of ports, $N^*$, increases with the increase in the number of ports for which a UE can skip their channel estimation processes. Another interesting observation is that, beyond the scenario of FAs with small number of ports (e.g., $1\leq N\leq 7$), the employment of the proposed SeCE technique leads to the improvement of the achieved network performance compared to the conventional Skip-less CE technique (i.e., $\nu=0$). This was expected since, the increased number of ports for which the channel estimation process can be skipped, triggers the allocation of more time for pilot-based training for the remaining ports, enhancing the channel estimation quality, and therefore boosting the outage performance. In particular, the employment of the proposed technique with $\nu=3$, results in an increase of the achieved outage performance by $58\%$ compared with the conventional technique for the scenario where all FAs are equipped with $N=30$ ports. For comparison purposes, we also present the outage performance obtained with a perfect (a-priori) CSI \cite{WON3}, denoted as ``Perfect CSI''. We can easily observe that, in contrast to the scenario considered with channel estimation error, the network performance with a perfect CSI is constantly increasing with the increase of FA ports. This is due to the fact that the negative effect of channel estimation quality on the network's performance is neglected.
\begin{figure}
	\centering\includegraphics[width=.6\linewidth]{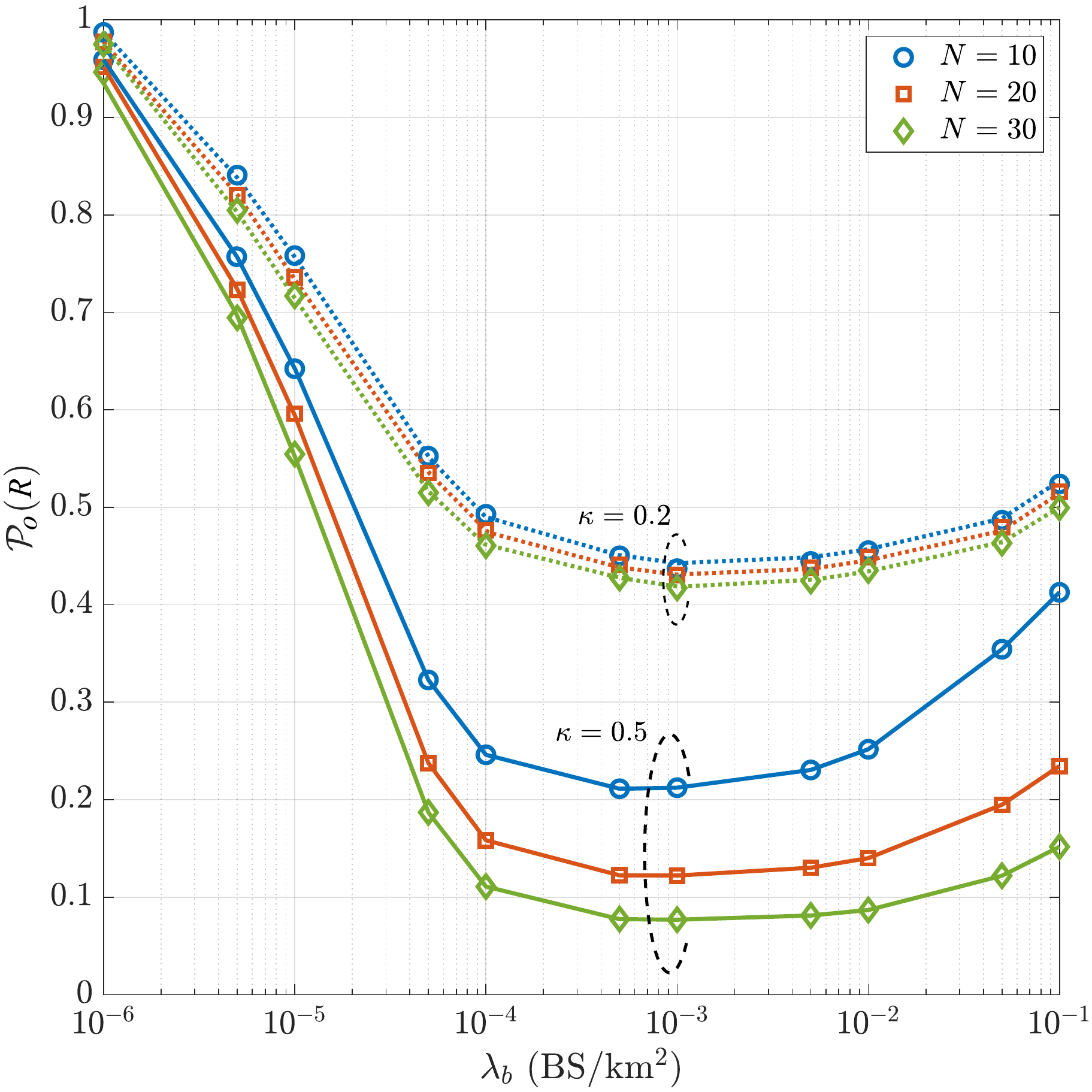}
	\caption{Outage performance versus the BSs density $\lambda_b$ (BS/km$^2$) for different $N$ and $\kappa$.}\label{ThirdFigure}\vspace{-20pt}
\end{figure}

Fig. \ref{ThirdFigure} evaluates the outage performance with respect to the density of BSs for different scaling constants $\kappa=\{0.2,0.5\}$. As mentioned before, the experienced outage performance of a FA-based UE can be reduced by adopting larger FA architectures i.e., a larger $\kappa$. Another interesting observation is that the outage performance initially decreases with $\lambda_b$ but, after a certain value of $\lambda_b$, it starts to increase. This observation is based on the fact that at low density values, the increased number of BSs leads to the reduction of the distance between the UEs and their serving BSs, thereby the observed SINR at the UEs is enhanced. On the other hand, for high BSs' densities, the overall interference caused by the active UEs increases, thereby reducing the ability of the BSs to decode the received signal. Finally, Fig. \ref{ThirdFigure} reveals that the negative impact of the scaling constant on the performance experienced by the FA-enabled UEs strongly depends on the spatial density of the BSs. In particular, for dense network deployments (i.e. high density values), by increasing the size of the UEs' FAs, that is achieved by increasing the scaling constant, leads to a greater reduction in the outage performance compared to that observed in sparse network deployments (i.e. low density values).

Fig. \ref{LastFigure} illustrates the effect of the fluid metal's speed on the number of FAs for which their channel estimation processes must skipped in order to ensure a desired channel estimation quality. More specifically, we plot the minimum number of skipped FA ports (derived in Proposition \ref{Proposition2}), $\nu^*$, versus the desired variance of channel estimation error, $\varsigma$, for different $N=\{20,40\}$ and $\Delta\phi=\{1,10,100\} V$. We can easily observe from the figure that, in order to obtain enhanced channel estimation quality i.e., $\varsigma\rightarrow 0$, a larger number of skipped FAs' ports is required. This was expected since, by increasing the number of FAs' ports for which their channel estimation processes can be omitted, leads to the allocation of more channel uses, and thus pilot-training symbols, for all ``selected'' ports, and consequently, the channel estimation quality is significantly improves. On the other hand, as the desired channel estimation quality decreases i.e.,$\varsigma\rightarrow 1$, the channel estimation process of more and more FAs' ports can be performed, decreasing the minimum number of skipped FAs' ports. Moreover, based on the figure, we can observe the positive effect of the number of FAs' ports on the minimum number of skipped FAs' ports. As expected, the increased number of FAs' ports leads to the demand of omitting of a larger number of estimated channel coefficients in order to maintain the same allocation of pilot-training symbols for all ``selected'' ports. Finally, Fig. \ref{LastFigure} highlights the importance of considering the delay occurred by the movement of fluid metal between the FAs' ports, owing to its finite velocity within the capillary of a FA. 
\begin{figure}
	\centering\includegraphics[width=.6\linewidth]{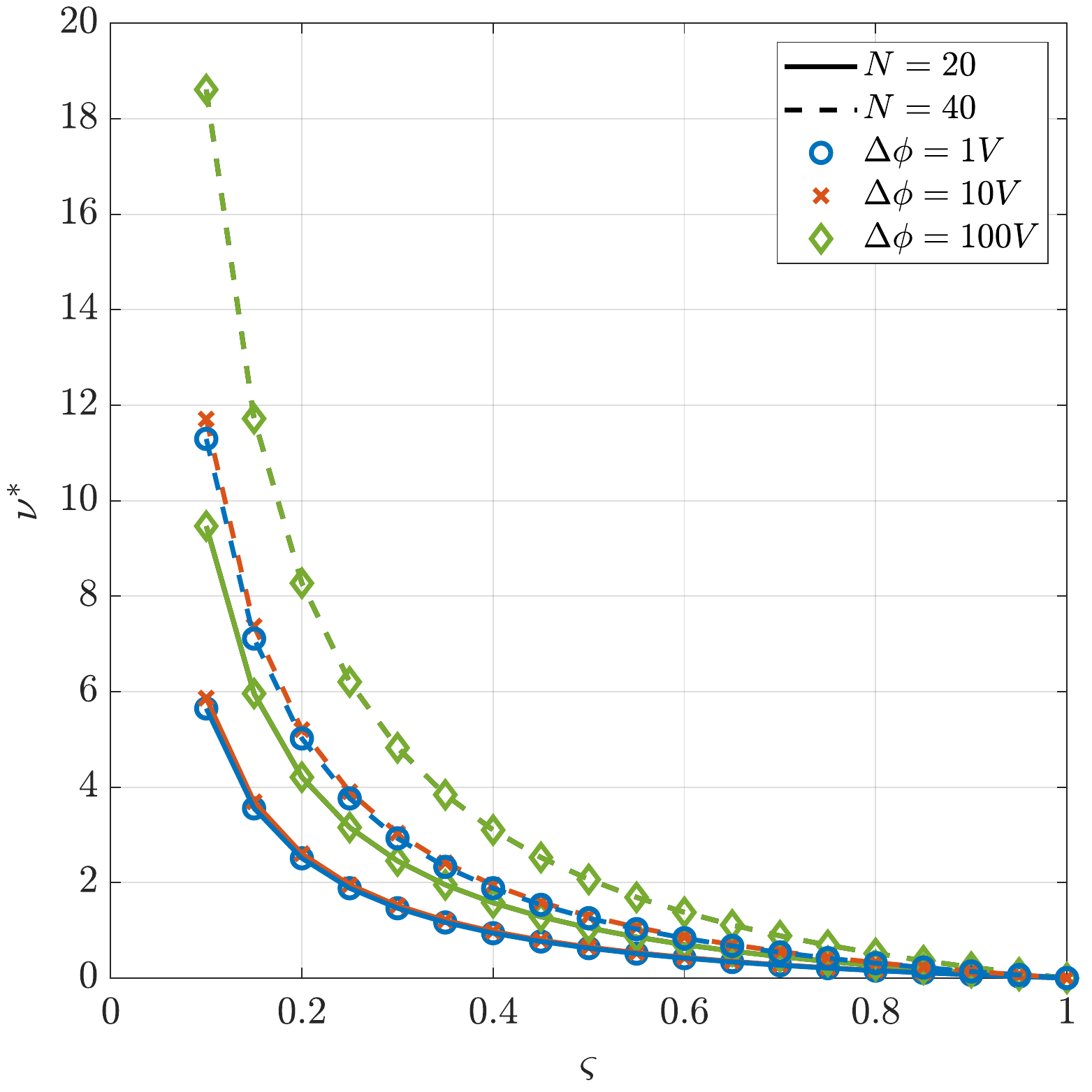}
	\caption{Minimum number of skipped FAs' ports versus the targeted variance of channel estimation error $\varsigma$ for different $N$ and $\Delta\phi$.}\label{LastFigure}\vspace{-20pt}
\end{figure} 
\section{Conclusion}\label{Conclusion}
In this paper, we presented an analytical framework based on SG to study the outage performance of FA-based UEs in the context of large-scale homogeneous cellular networks. The developed framework takes into account the employment of a circular multi-FA array at the UEs, as well as the presence of both channel estimation error and channel correlation effects. Aiming at improving the channel estimation quality and thereby enhancing the network's performance in practical limited coherence interval scenarios, we proposed a low-complexity PS scheme that employs a novel LMMSE-based channel estimation technique. In particular, the proposed technique activates a single port of a UE's FA array based on the observed SINR, among a subset of FA ports that are estimated to provide the strongest channel. Analytical and approximated closed-form expressions for the outage performance in the context of our proposed technique were derived, and the impact of nodes density, block length, and number of FAs' ports has been discussed. Our results highlight the impact of the FAs' architecture and the network topology on the optimal number of FA and ports, providing guidance for the planning of cellular networks in order to achieve enhanced network performance. Finally, our results demonstrated that the proposed scheme provides significant performance gains compared to the conventional case, whilst keeping the complexity low.

%For the purpose of addressing the new and inevitable challenges induced by the FA technology, in this paper, we have explored both a novel channel estimation scheme and a low complexity port selection mechanism in the context of FA-enabled wireless cellular networks. To cope with the randomness of the considered network deployments, we have proposed an analytical framework based on SG to study the outage performance of UEs, that employ a circular multi-FA array, from a macroscopic point-of-view. The developed mathematical framework takes into account the spatial randomness of BSs and UEs, as well as the presence of both channel estimation error and channel correlation effects. Based on the developed mathematical framework, the outage performance achieved with the proposed techniques was analytically derived and the impact of nodes density, block length, and number of FAs' ports has been discussed. Our results highlight the impact of the FAs' architecture and the network topology on the optimal number of FA and ports, providing guidance for the planning of cellular networks in order to achieve enhanced network performance. 

\appendices
\section{Proof of Proposition \ref{Proposition1}}\label{Appendix1}
By leveraging the orthogonality property of the LMMSE estimator, we have $\mathbb{E}\left[\hat{g}_{0i}^{(k)}|_{\rho}\left(e_i|_{\rho}\right)^*\right]=0$, where $g_{0i}^{(k)}=\hat{g}_{0i}^{(k)}|_{\rho}+e_i|_{\rho}$. Hence, the estimation variance is given by
\begin{equation}
\sigma^2_{e}|_{\rho}\triangleq\mathbb{E}\left[\left(e_i|_{\rho}\right)^2\right]=1-\mathbb{E}\left[\left|\hat{g}_{0i}^{(k)}|_{\rho} \right|^2 \right],
\end{equation}
where, by using \eqref{VarianceProof}
\begin{equation}
\mathbb{E}\left[\left|\hat{g}_{0i}^{(k)}|_{\rho} \right|^2 \right] = \left(\frac{\sqrt{\frac{L_e P}{\ell(r_i(\rho))}}}{\frac{L_e P}{\ell(r_i(\rho))}+\frac{N_0}{\sigma^2}+P\mathbb{E}(\mathcal{I})}\right)^2\mathbb{E}\left[\left(\widetilde{y}_i^{(k)}\right)^*\widetilde{y}_i^{(k)}|\rho\right].
\end{equation}
The final expression follows by evaluating the expectation.

\section{Proof of Lemma \ref{Lemma1}}\label{Appendix2}
Under the proposed LMMSE-based SeCE technique, the channel experienced by the $i$-th port of the typical UE's $k$-th FA can be expressed as $\widehat{g}_{{\rm 0}i}^{(k)}|_{\rho}=g_{{\rm 0}i}^{(k)}+e_i|_{\rho}$, where $g_{{\rm 0}i}^{(k)}\stackrel{d}{\sim}\mathcal{CN}(0,\sigma^2)$ and $e_i|_{\rho}\stackrel{d}{\sim}\mathcal{CN}\left(0,\sigma^2_{e}|_{\rho}\right)$, and therefore, $\widehat{g}_{{\rm 0}i}^{(k)}|_{\rho}\stackrel{d}{\sim}\mathcal{CN}\left(0,\widetilde{\sigma}_i^2\right)$, where $\widetilde{\sigma}_i^2=\sigma^2(1-\mu_i^2)+\sigma^2_{e}|_{\rho}$. Under this model, the amplitude of the estimated channels, $|\widehat{g}_{{\rm 0}i}^{(k)}|$, is Rayleigh distributed, with pdf 
\begin{equation}\label{pdf}
f_{|\widehat{g}_{{\rm 0}i}^{(k)}|}(\tau) = \frac{2\tau}{\widetilde{\sigma}_i^2}\exp\left(-\frac{\tau^2}{\widetilde{\sigma}_i^2}\right),
\end{equation}
with $\mathbb{E}\left[|\widehat{g}_{{\rm 0}i}^{(k)}|^2\right]=\widetilde{\sigma}_i$. As it already mentioned, the channels $\{\widehat{g}_{{\rm 0}i}^{(k)}\}$ are correlated due to the capability of FA's ports to be arbitrarily close to each other. In particular, the amplitude of the estimated channel $|\widehat{g}_{{\rm 0}2}^{(k)}|$, conditioned on $\rho$, $x_0$ and $y_0$, follows a Rice distribution, i.e.,
\begin{equation}
f_{\left|\widehat{g}_{{\rm 0}2}^{(k)}\right||\{x_0,y_0\}}(\tau_2|\rho,x_0,y_0)=\frac{2\tau_2}{\widetilde{\sigma}_i^2}\exp\left(-\frac{\tau_2^2+\mu^2_2(x_0^2+y_0^2)}{\widetilde{\sigma}_i^2}\right)I_0\left(\frac{2\mu_2\sqrt{x_0^2+y_0^2}\tau_2}{\widetilde{\sigma}_i^2}\right),
\end{equation}
where $\tau_2\geq 0$. By substituting $\tau_1\rightarrow\sqrt{x_0^2+y_0^2}$ and since $x_0,y_0,\left|\widehat{g}_{{\rm 0}2}^{(k)}\right|,\dots,\left|\widehat{g}_{{\rm 0}N'}^{(k)}\right|$ are all independent between each other, the joint pdf of the estimated channels, conditioned on $\left|\widehat{g}_{{\rm 0}1}^{(k)}\right|$, can be expressed as
\begin{equation}
f_{\left|\widehat{g}_{{\rm 0}2}^{(k)}\right|,\dots,\left|\widehat{g}_{{\rm 0}N'}^{(k)}\right||\left|\widehat{g}_{{\rm 0}1}^{(k)}\right|}(\tau_2,\dots,\tau_{N'}|\rho,\tau_1)=\prod_{i\in\mathfrak{N}}\frac{2\tau_i}{\widetilde{\sigma}_i^2}\exp\left(-\frac{\tau_i^2+\mu^i_2\tau_1^2}{\widetilde{\sigma}_i^2}\right)I_0\left(\frac{2\mu_i\tau_1\tau_i}{\widetilde{\sigma}_i^2}\right).
\end{equation}
Then, the final expression can be achieved by un-conditioning the above expression with the pdf of $\left|\widehat{g}_{{\rm 0}1}^{(k)}\right|$ given in \eqref{pdf}, i.e. $f_{\left|\widehat{g}_{{\rm 0}2}^{(k)}\right|,\dots,\left|\widehat{g}_{{\rm 0}N'}^{(k)}\right||\left|\widehat{g}_{{\rm 0}1}^{(k)}\right|}(\tau_2,\dots,\tau_{N'}|\rho,\tau_1)f_{\left|\widehat{g}_{{\rm 0}1}^{(k)}\right|}(\tau_1)$, which gives the desired expression.

Regarding the joint cdf of $\left|\widehat{g}_{{\rm 0}i}^{(k)}\right|$, based on the definition, this is given by
\begin{equation}
F_{\left|\widehat{g}_{{\rm 0}1}^{(k)}\right|,\dots,\left|\widehat{g}_{{\rm 0}N'}^{(k)}\right|}(\tau_1,\dots,\tau_{N'}|\rho)=\int_{0}^{\tau_1}\cdots\int_{0}^{\tau_{N'}}f_{\left|\widehat{g}_{{\rm 0}1}^{(k)}\right|,\dots,\left|\widehat{g}_{{\rm 0}N'}^{(k)}\right|}(\tau_1,\dots,\tau_{N'}|\rho){\rm d}\tau_1\cdots{\rm d}\tau_{N'}.
\end{equation}
By using the derived expression for the joint pdf which is given by \eqref{PDF}, the above can be re-written as
\begin{align}
&F_{\left|\widehat{g}_{{\rm 0}1}^{(k)}\right|,\dots,\left|\widehat{g}_{{\rm 0}N'}^{(k)}\right|}(\tau_1,\dots,\tau_{N'}|\rho)\nonumber\\&\qquad=\int_0^{\tau_1}\frac{2t_1}{\widetilde{\sigma}_1^2}\exp\left(-\frac{t_1^2}{\widetilde{\sigma}_1^2}\right)\prod_{j\in\mathfrak{N}}\left[\int_0^{\tau_j}\frac{2t_j}{\widetilde{\sigma}_j^2}\exp\left(-\frac{t_j^2+\mu^2_jt_1^2}{\widetilde{\sigma}_j^2}\right)I_0\left(\frac{2\mu_jt_1t_j}{\widetilde{\sigma}_j^2}\right){\rm d}t_j\right]{\rm d}t_1.
\end{align}
Note that, the integral inside the product operator is an integration over the pdf of a Ricean random variable, which therefore gives \cite[2.20]{Marcumbook}
\begin{equation*}
F_{\left|\widehat{g}_{{\rm 0}1}^{(k)}\right|,\dots,\left|\widehat{g}_{{\rm 0}N'}^{(k)}\right|}(\tau_1,\dots,\tau_{N'}|\rho)=\int_0^{\tau_1}\frac{2t_1}{\widetilde{\sigma}_1^2}\exp\left(-\frac{t_1^2}{\widetilde{\sigma}_1^2}\right)\prod_{j\in\mathfrak{N}}\left[1-Q_1\left(\sqrt{\frac{2\mu_j^2}{\widetilde{\sigma}_j^2}}t_1,\sqrt{\frac{2}{\widetilde{\sigma}_j^2}}\tau_j\right)\right]{\rm d}t_1.
\end{equation*}
The final expression can be derived by using the transformation $t=\frac{t_1^2}{\widetilde{\sigma}_1^2}$, which concludes the proof.

\section{Proof of Proposition \ref{InterferenceDistribution}}\label{Appendix3}
The multi-user interference can be approximated by matching the mean and variance of the Gamma distribution with $\mathbb{E}[\mathcal{I}]$ and ${\rm Var}(\mathcal{I})$, resulting in a very tight approximation. In particular, the mean and the variance of the multi-user interference can be calculated as in \eqref{MeanInterference} and 
\begin{equation}
{\rm Var}(\mathcal{I})=\int_0^\infty \frac{2\tau_1}{\sigma^2}\exp\left(-\frac{\tau_1^2}{\sigma^2}\left(\int_{0}^\infty\tau_i^4f_{\tau_i|\tau_1}(\tau_i|\tau_1)\right){\rm d}\tau_i\right){\rm d}\tau_1=2\sigma^4,
\end{equation}
respectively, where $f_{\tau_i|\tau_1}(\tau_i|\tau_1)$ can be obtained by \cite[Theorem 1]{WON1}. The pdf of a random variable $X$ that follows a Gamma distribution with shape parameters 
$k,\ \theta>0$, is given by \eqref{Gammapdf}, with the first and second moments of random variable $X$ are given by $\mathbb{E}[X]=k\theta$ and $\mathbb{E}[X^2]=k\theta^2$, respectively. Hence, letting $\mathbb{E}[X]=\mathbb{E}[\mathcal{I}]$ and $\mathbb{E}[X^2]={\rm Var}(\mathcal{I})$, the shape parameters of the Gamma distribution can be derived by setting $k=(\mathbb{E}[\mathcal{I}])^2/{\rm Var}(\mathcal{I})$ and $\theta={\rm Var}(\mathcal{I})/\mathbb{E}[\mathcal{I}]$, respectively.

\section{Proof of Lemma \ref{Lemma3}}\label{Appendix4}
By definition, $0\leq Q_1(\alpha,\beta)\leq 1$, and therefore $\prod_j(1-Q_1(\alpha,\beta))\approx 1-\sum_jQ_1(\alpha,\beta)$. Then, for the special case where $\mu_2=\dots=\mu_N=\mu$, the expression for the conditional outage probability can be approximated as follows
\begin{align}
	\mathcal{P}_o(R|\rho,\mathcal{I}_{\mathtt{i}_k})& \approx  \int_0^{\frac{\Theta_1^2}{\widetilde{\sigma}_1^2}}\exp(-t)\left[1-\sum_{j\in\mathfrak{N}}Q_1\left(\sqrt{2\mu^2\frac{\widetilde{\sigma}_1^2}{\widetilde{\sigma}_j^2}t},\sqrt{\frac{2}{\widetilde{\sigma}_j^2}}\Theta_j\right)\right]{\rm d}t\nonumber \\
	&=1-\exp\left(-\frac{\Theta_1^2}{\widetilde{\sigma}_1^2}\right)-\sum_{j\in\mathfrak{N}}\int_0^{\frac{\Theta_1^2}{\widetilde{\sigma}_1^2}}\exp(-t)Q_1\left(\sqrt{2\mu^2\frac{\widetilde{\sigma}_1^2}{\widetilde{\sigma}_j^2}t},\sqrt{\frac{2}{\widetilde{\sigma}_j^2}}\Theta_j\right){\rm d}t,\label{Proof1}
\end{align}
where $\Theta_j \stackrel{\epsilon\rightarrow\infty}{\approxeq} \vartheta \left(r_j^a(\rho)\mathcal{I}_j+\sigma^2_e|_{\rho}\right)$, and \eqref{Proof1} follows by changing the order of the summation and the integral. Due to the small FAs' form factor, and thus, the relatively smaller distance between FA's ports compared to their communication link length with the serving BS, we assume that all UEs' ports are equidistant with the serving BS i.e., $r_i(\rho)\approxeq r(\rho)$ $\forall i\in\mathfrak{N}$ when $\lambda_b\ll 1$. In the context of the considered interference-limited scenario with $\mu_2=\dots=\mu_N=\mu$, all FAs' ports experience the same variance of channel estimation error i.e., $\sigma^2_{e}|_{\rho} \approxeq2\pi\lambda_b\frac{N}{L_T}\frac{r(\rho)^{2}}{a-2}$, and hence $\widetilde{\sigma}_j^2=\widetilde{\sigma}^2,\ \forall j\in\mathfrak{N}\backslash\{1\}$. Then, based on \cite[C.23]{Marcumbook}, the following upper bound of $Q_1(\alpha,\beta)$ exists
\begin{equation}\label{bounds}
	\exp\left(-\frac{(\beta+\alpha)^2}{2}\right)\leq Q_1(\alpha,\beta)\leq \exp\left(-\frac{(\beta-\alpha)^2}{2}\right).
\end{equation}
Using the above upper bound on the $Q_1(\cdot, \cdot)$ term inside the integration of \eqref{Proof1} gives
\begin{align}
	\mathcal{P}_o(R|\rho,\mathcal{I}_{\mathtt{i}_k}) &\leq 1-\exp\left(-\frac{\Theta_1^2}{\widetilde{\sigma}^2}\right)-\sum_{j\in\mathfrak{N}}\int_0^{\frac{\Theta_1^2}{\widetilde{\sigma}^2}}\exp(-t)\exp\left(-\frac{\left(\sqrt{\frac{2}{\widetilde{\sigma}^2}}\Theta_j-\sqrt{2\mu^2t}\right)^2}{2}\right){\rm d}t\nonumber\\
	&\!=\! 1\!-\!\exp\left(\!-\frac{\Theta_1^2}{\widetilde{\sigma}^2}\right)\!-\!\sum_{j\in\mathcal{N}}\frac{\exp\left(-\frac{\Theta_j^2}{\widetilde{\sigma}^2}\right)}{1+\mu^2}\Bigg(1-\exp\left(2\sqrt{\mu^2\frac{\Theta_1^2\Theta_j^2}{\widetilde{\sigma}^4}}-\frac{\Theta_1^2}{\widetilde{\sigma}^2}(1+\mu^2)\right)\nonumber\\
	&\quad\!+\!\sqrt{\!\frac{\pi\mu^2}{1\!+\!\mu^2}\frac{\Theta_j^2}{\widetilde{\sigma}^2}}\exp\!\left(-\frac{\Theta_j^2}{\widetilde{\sigma}^2(1\!+\!\mu^2)}\right)\!\Delta\!\left(\!\sqrt{\frac{\Theta_1^2}{\widetilde{\sigma}^2}(1\!+\!\mu^2)},\sqrt{\frac{\Theta_j^2}{\widetilde{\sigma}^2}\frac{\mu^2}{1\!+\!\mu^2}}\right)\!\Bigg),\label{Proof2}
\end{align}
where \eqref{Proof2} follows based on \cite[C.23]{Marcumbook}, $\Delta(\alpha,\beta)= {\rm erf}[\alpha-\beta]+{\rm erf}[\beta]$, and ${\rm erf}[\cdot]$ is the error function i.e., ${\rm erf}[\alpha]=\frac{2}{\sqrt{\pi}}\int_0^\alpha e^{-t^2}{\rm d}t$. Since $0\leq\mu\leq 1$, the inequality $1+\mu^2>\frac{\mu^2}{1+\mu^2}$ holds, and therefore, $\Delta(\alpha,\beta)\approxeq2$. Hence, by letting $\Xi_i = \frac{\Theta_i^2}{\widetilde{\sigma}^2}$ and by assuming $\frac{\Xi_j}{\Xi_1}\rightarrow1$ when $\Xi_i\rightarrow\infty,\ \forall i\in\mathfrak{N}$, the final expression can be derived. Regarding the lower bound for the conditional outage probability, based on \eqref{bounds}, $\mathcal{P}_o(R|\rho,\mathcal{I}_{\mathtt{i}_k})$ is lower bounded as follows
\begin{align}
	\mathcal{P}_o(R|\rho,\mathcal{I}_{\mathtt{i}_k}) &\geq 1-\exp\left(-\frac{\Theta_1^2}{\widetilde{\sigma}^2}\right)-\sum_{j\in\mathfrak{N}}\int_0^{\frac{\Theta_1^2}{\widetilde{\sigma}^2}}\exp(-t)\exp\left(-\frac{\left(\sqrt{\frac{2}{\widetilde{\sigma}^2}}\Theta_j+\sqrt{2\mu^2t}\right)^2}{2}\right){\rm d}t\nonumber\\
	&\!=\! 1\!-\!\exp\left(\!-\frac{\Theta_1^2}{\widetilde{\sigma}^2}\right)\!-\!\sum_{j\in\mathfrak{N}}\frac{\exp\left(-\frac{\Theta_j^2}{\widetilde{\sigma}^2}\right)}{1+\mu^2}\Bigg(1-\exp\left(-2\sqrt{\mu^2\frac{\Theta_1^2\Theta_j^2}{\widetilde{\sigma}^4}}-\frac{\Theta_1^2}{\widetilde{\sigma}^2}(1+\mu^2)\right)\nonumber\\
	&\quad\!-\!\sqrt{\!\frac{\pi\mu^2}{1\!+\!\mu^2}\frac{\Theta_j^2}{\widetilde{\sigma}^2}}\exp\!\left(-\frac{\Theta_j^2}{\widetilde{\sigma}^2(1\!+\!\mu^2)}\right)\!\widetilde{\Delta}\!\left(\!\sqrt{\frac{\Theta_1^2}{\widetilde{\sigma}^2}(1\!+\!\mu^2)},\sqrt{\frac{\Theta_j^2}{\widetilde{\sigma}^2}\frac{\mu^2}{1\!+\!\mu^2}}\right)\!\Bigg),\label{Proof3}
\end{align} 
where $\widetilde{\Delta}(\alpha,\beta)= {\rm erf}[\alpha+\beta]-{\rm erf}[\beta]$. Note that, $\widetilde{\Delta}(\alpha,\beta)\approxeq0$ for $1+\mu^2>\frac{\mu^2}{1+\mu^2}$. Hence, by following similar approach as with the derivation of the lower bound and according to \cite[C.23]{Marcumbook}, the lower bound of the conditional outage probability can be derived.

\end{document}